\documentclass[10pt, twoside]{IEEEtran} 
\usepackage{epsfig,algorithm,amsmath,amssymb,color,url,algorithmic,amsthm}
\usepackage{slashbox}
\usepackage{float}
\usepackage{epstopdf}
\usepackage{caption}
\usepackage{subcaption}
\usepackage{bm}
\newtheorem{theorem}{Theorem}[section]
\newtheorem{assumption}[theorem]{Assumption}
\newtheorem{lemma}[theorem]{Lemma}
\newtheorem{corollary}[theorem]{Corollary}
\newtheorem{proposition}[theorem]{Proposition}
\newtheorem{definition}[theorem]{Definition}

\newtheorem{remark}[theorem]{Remark}

\newcommand{\<}{\langle}
\renewcommand{\>}{\rangle}
\newcommand{\new}{{\text{new}}}
\newcommand{\extra}{{\text{extra}}}
\newcommand{\Span}{{\text{range}}}
\newcommand{\rank}{{\text{rank}}}
\newcommand{\Ber}{{\text{Ber}}}
\newcommand{\Unif}{{\text{Unif}}}
\newcommand{\Succ}{{\text{Success}}}
\newcommand{\trace}{{\text{trace}}}
\newcommand{\old}{\text{old}}
\newcommand{\sgn}{\text{sgn}}
\newcommand{\R}{\mathbb{R}}
\newcommand{\bOmega}{\bar{\Omega}}
\newcommand{\tL}{\tilde{\Lb}}
\newcommand{\tS}{\tilde{\Sb}}
\newcommand{\tX}{\tilde{\X}}

\newcommand{\N}{\mathcal{N}}
\newcommand{\E}{\mathbb{E}}
\newcommand{\I}{\mathcal{I}}
\newcommand{\Pc}{\mathcal{P}}
\newcommand{\Pb}{\mathbb{P}}
\newcommand{\PO}{\mathcal{P}_{\Omega}}
\newcommand{\PT}{\mathcal{P}_{T_\new}}

\newcommand{\PGp}{\mathcal{P}_{\Gamma^\perp}}
\newcommand{\PP}{\mathcal{P}_\Pi}
\newcommand{\PPp}{\mathcal{P}_{\Pi^\perp}}

\newcommand{\POp}{\mathcal{P}_{\Omega^\perp}}
\newcommand{\POj}{\mathcal{P}_{\bOmega_j}}
\newcommand{\POzero}{\mathcal{P}_{\Omega_0}}

\newcommand{\A}{\mathbf{A}}
\newcommand{\B}{\mathbf{B}}
\newcommand{\D}{\mathbf{D}}
\newcommand{\Eb}{\mathbf{E}}
\newcommand{\X}{\mathbf{X}}
\newcommand{\Y}{\mathbf{Y}}
\newcommand{\Z}{\mathbf{Z}}
\newcommand{\Sb}{\mathbf{S}}
\newcommand{\Lb}{\mathbf{L}}
\newcommand{\M}{\mathbf{M}}
\newcommand{\G}{\mathbf{G}}
\newcommand{\Hb}{\mathbf{H}}
\newcommand{\Pbf}{\mathbf{P}}
\newcommand{\Q}{\mathbf{Q}}
\newcommand{\U}{\mathbf{U}}
\newcommand{\V}{\mathbf{V}}
\newcommand{\W}{\mathbf{W}}
\newcommand{\Rb}{\mathbf{R}}
\newcommand{\F}{\mathbf{F}}
\newcommand{\Sigmab}{\mathbf{\Sigma}}
\newcommand{\Ib}{\mathbf{I}}
\newcommand{\ab}{\mathbf{a}}
\newcommand{\e}{\mathbf{e}}
\newcommand{\s}{\mathbf{s}}
\newcommand{\x}{\mathbf{x}}
\newcommand{\y}{\mathbf{y}}
\newcommand{\lb}{\bm{\ell}}
\newcommand{\zero}{\mathbf{0}}

\allowdisplaybreaks
\begin{document}

\title{Robust PCA with Partial Subspace Knowledge}

\author{
  \IEEEauthorblockN{Jinchun Zhan and Namrata Vaswani \\}
  \IEEEauthorblockA{ECE dept, Iowa State University, Ames, Iowa, USA\\
    Email: jzhan@iastate.edu, namrata@iastate.edu}
\thanks{A shorter version of this paper appears in the proceedings of ISIT 2014 \cite{modpcp_isit}. This work was supported in part by NSF grant CCF-1117125}
}
\maketitle

\begin{abstract}
In recent work, robust Principal Components Analysis (PCA) has been posed as a problem of recovering a low-rank matrix $\Lb$ and a sparse matrix $\Sb$ from their sum, $\M:= \Lb + \Sb$ and a provably exact convex optimization solution called PCP has been proposed. This work studies the following problem. Suppose that we have partial knowledge about the column space of the low rank matrix $\Lb$. Can we use this information to improve the PCP solution, i.e. allow recovery under weaker assumptions? We propose here a simple but useful modification of the PCP idea, called modified-PCP, that allows us to use this knowledge. We derive its correctness result which shows that, when the available subspace knowledge is accurate, modified-PCP indeed requires significantly weaker incoherence assumptions than PCP. Extensive simulations are also used to illustrate this. Comparisons with PCP and other existing work are shown for a stylized real application as well. Finally, we explain how this problem naturally occurs in many applications involving time series data, i.e. in what is called the online or recursive robust PCA problem. A corollary for this case is also given.
\end{abstract}

\section{Introduction}
Principal Components Analysis (PCA) is a widely used dimension reduction technique that finds a small number of orthogonal
basis vectors, called principal components, along which most of the variability of the dataset lies. Accurately computing the principal components in the presence of outliers is called robust PCA. Outlier is a loosely defined term that refers to any corruption that is not small compared to the true data vector and that occurs occasionally. As suggested
in \cite{error_correction_PCP_l1}, an outlier can be nicely modeled as a sparse vector.
The robust PCA problem occurs in various applications ranging from video analysis to recommender system design in the presence of outliers, e.g. for Netflix movies, to anomaly detection in dynamic networks \cite{rpca}. In recent work, Candes et al and Chandrasekharan et al \cite{rpca,rpca2} posed the robust PCA problem as one of separating a low-rank matrix $\Lb$ (true data matrix)  and a sparse matrix $\Sb$ (outliers' matrix) from their sum, $\M:= \Lb + \Sb$. They showed that by solving the following convex optimization program
\begin{equation}
\label{eq:pcp}
  \begin{array}{ll}
    \text{minimize}_{\tL,\tS}   & \quad \|\tL\|_* + \lambda \|\tS\|_1\\
    \text{subject to} & \quad \tL + \tS =  \M
  \end{array}
\end{equation}
it is possible to recover $\Lb$ and $\Sb$ exactly with high probability (w.h.p.) under mild assumptions. In \cite{rpca}, they called it principal components' pursuit (PCP).
Here $\|\tL\|_*$ denotes the nuclear norm of $\tL$ and $\|\tS\|_1$ denotes the $\ell_1$ norm of $\tS$ reshaped as a long vector. 
This was among the first recovery guarantees for a practical (polynomial complexity) robust PCA algorithm. Since then, the batch robust PCA problem, or what is now also often called the sparse+low-rank recovery problem, has been studied extensively but theoretically and empirically, e.g. see \cite{error_correction_PCP_l1,novel_m_estimator,mccoy_tropp11, outlier_pursuit,RMD,noisy_undersampled_yuan,agarwal2012noisy,seidel2013prost,xu2013gosus,bouwmans2014robust}.


{\em Contribution: }
In this work we study the following problem. Suppose that we have a partial estimate of the column space of the low rank matrix $\Lb$. How can we use this information to improve the PCP solution, i.e. allow recovery under weaker assumptions? We propose here a simple but useful modification of the PCP idea, called {\em modified-PCP}, that allows us to use this knowledge. We derive its correctness result (Theorem \ref{maintheorem}) that provides explicit bounds on the various constants and on the matrix size that are needed to ensure exact recovery with high probability. Our result is used to argue that, as long as the  available subspace knowledge is accurate, modified-PCP requires significantly weaker incoherence assumptions than PCP. To prove the result, we use the overall proof approach of \cite{rpca} with some changes (explained in Sec \ref{proof_maintheorem}). By ``accurate" subspace knowledge, we mean that the number of missed directions and the number of extra directions in the available subspace knowledge is small compared to the rank of $\Lb$.

An important problem where partial subspace knowledge is available is in online or recursive robust PCA for sequentially arriving time series data, e.g. for video based foreground and background separation.  Video background sequences are well modeled as forming a low-rank but dense matrix because they change slowly over time and the changes are typically global. Foreground is a sparse image consisting of one or more moving objects. As explained in \cite{rrpcp_perf}, in this case, the subspace spanned by a set of consecutive columns of $\Lb$ does not remain fixed, but instead changes gradually over time. Also, often an initial short sequence of low-rank only data (without outliers) is available, e.g. in video analysis, it is easy to get an initial background-only sequence. For this application, modified-PCP can be used to design a piecewise batch solution that will be faster and will require weaker assumptions for exact recovery than PCP. This is made precise in Corollary \ref{cor}.

We also show extensive simulation comparisons and some real data comparisons of modified-PCP with PCP and with other existing robust PCA solutions from literature. The implementation requires a fast algorithm for solving the modified-PCP program. We develop this  by modifying the Inexact Augmented Lagrange Multiplier Method of \cite{alm} and using the idea of \cite{hale2008fixed,cai2010singular} for the sparse recovery step.


{\bf Notation. }
For a matrix $\X$, we denote by $\X^*$ the transpose of $\X$; denote by $\|\X\|_{\infty}$ the $\ell_{\infty}$ norm of $\X$ reshaped as a long vector, i.e., $\max_{i,j} |\X_{ij}|$; denote by $\|\X\|$ the operator norm or 2-norm; denote by $\|\X\|_F$ the Frobenius norm. 

Let $\I$ denote the identity operator, i.e., $\I(\Y)=\Y$ for any matrix $\Y$. Let $\|\mathcal{A}\|$ denote the operator norm of operator $\mathcal{A}$, i.e., $\|\mathcal{A}\| = \sup_{\{\|\X\|_F = 1\}} \|\mathcal{A} \X\|_F$; let $\langle \X, \Y\rangle$ denote the Euclidean inner product between two matrices, i.e., trace($\X^*\Y$); let sgn$(\X)$ denote the entrywise sign of $\X$.

We let ${\cal P}_\Theta$ denote the orthogonal projection onto a linear subspace $\Theta$ of matrices.
We use $\Omega$ to denote the support set of $\Sb$, i.e., $\Omega=\{(i,j): \Sb(i,j)\neq 0\}$. As is done in \cite{rpca}, we also use $\Omega$ to denote the subspace spanned by the matrices supported on the set $\Omega$ (i.e. matrices whose entries are zero on the complement of the set $\Omega$).
For a matrix $\X$, we use $\PO \X$ to denote projection onto the subspace $\Omega$, i.e., $(\PO \X)_{ij}=\X_{ij},$ if $(i,j)\in \Omega$, and $(\PO \X)_{ij}=0,$ if $(i,j)\notin \Omega$. 
By $\Omega \sim \Ber(\rho)$ we mean that any matrix index $(i,j)$ has probability $\rho$ of being in the support independent of all others.

Given two matrices $\B$ and $\B_2$, $[\B \ \B_2]$ constructs a new matrix by concatenating matrices $\B$ and $\B_2$ in the horizontal direction. Let $\B_{\text{rem}}$ be a matrix containing some columns of $\B$. Then $\B \setminus \B_{\text{rem}}$ is the matrix $\B$ with columns in $\B_{\text{rem}}$ removed.

We say that $\U$ is a {\em basis matrix} if $\U^*\U=\Ib$ where $\Ib$ is the identity matrix. We use  $\e_i$ to refer to the $i^{th}$ column $\Ib$. For a matrix $\U$, we use $\Span(\U)$ to denote its column span.

\section{Problem definition and proposed solution}

\subsection{Problem Definition} \label{probdef}
We are given a data matrix $\M \in \mathbb{R}^{n_1\times n_2}$ that satisfies
\begin{equation} \label{eq:M}
\M=\Lb+\Sb
\end{equation}
where $\Sb$ is a sparse matrix with support set $\Omega$ and $\Lb$ is a low rank matrix with reduced singular value decomposition (SVD)
\begin{equation}
\Lb \stackrel{\text{SVD}}{=} \U \Sigmab \V^*
\end{equation}

Let $r:= \rank(\Lb)$. We assume that we are given a basis matrix $\G$ so that $(\Ib-\G\G^*)\Lb$ has rank smaller than $r$. The goal is to recover $\Lb$ and $\Sb$ from $\M$ using $\G$. Let $r_G:=\rank(\G)$.

Define $\Lb_\new:= (\Ib-\G\G^*)\Lb$ with $r_\new: = \rank(\Lb_\new)$ and reduced SVD given by
\begin{equation} \label{eqVnew}
\Lb_\new:=(\Ib-\G\G^*)\Lb  \stackrel{\text{SVD}}{=} \U_\new \Sigmab_\new \V_\new^*
\end{equation}

%


We explain this a little more. With the above, it is easy to show that there exist rotation matrices $\Rb_U, \Rb_G$, and basis matrices $\G_\extra$ and $\U_\new$ with $\G_\extra{}^* \U_\new = 0$, such that
\begin{equation} \label{eqG}
\U = [\underbrace{(\G \Rb_G \setminus \G_\extra)}_{\U_0} \ \U_\new]\Rb_U^*.
\end{equation}
We provide a derivation for  this in Appendix \ref{probdef_proof}. Notice here that $\U_0$ be a basis matrix for $\Span(\Lb) \cap \Span(\G) = \Span(\U) \cap \Span(\G)$.

Define $r_0 := \rank(\U_0)$ and $r_\extra:= \rank(\G_\extra)$. Clearly, $r_G = r_0 + r_\extra$ and $r = r_0 + r_\new = (r_G - r_\extra) + r_\new$.





\subsection{Proposed Solution: Modified-PCP}

%
From the above model, it is clear that
\begin{equation}
\label{eq:GM}
\Lb_\new + \G\X^* + \Sb =  \M
\end{equation}
for $\X = \Lb^*\G$.
We propose to recover $\Lb$ and $\Sb$ using $\G$ by solving the following {\bf Modified PCP} (mod-PCP) program
\begin{equation}
\label{eq:sdp}
  \begin{array}{ll}
    \text{minimize}_{\tL_\new,\tS, \tX}   & \quad \|\tL_\new\|_* + \lambda \|\tS\|_1\\
    \text{subject to} & \quad \tL_\new + \G\tX^* + \tS =  \M
  \end{array}
\end{equation}
Denote a solution to the above by $\hat{\Lb}_\new,\hat{\Sb},\hat{\X}$. Then, $\Lb$ is recovered as $\hat{\Lb} = \hat{\Lb}_\new + \G \hat{\X}^*$. Modified-PCP is inspired by an approach for sparse recovery using partial support knowledge called modified-CS \cite{modcsjournal}.

\section{Correctness Result}
We first state the assumptions required for the result and then give the main result and discuss it.

\subsection{Assumptions} 
As explained in \cite{rpca}, we need that $\Sb$ is not low rank in order to separate it from $\Lb_\new$. One way to ensure that $\Sb$ is full rank w.h.p. is by selecting the support of $\Sb$ uniformly at random \cite{rpca}.  We assume this here too. In addition, we need a denseness assumption on $\G$ and on the left and right singular vectors of $\Lb_\new$. 

Let $n_{(1)} = \max(n_1,n_2)$ and $n_{(2)} = \min(n_1,n_2)$. Assume that following hold with a constant $\rho_r$ that is small enough (we set its values later in Assumption \ref{rhosr}). 
\begin{equation}
  \label{eq:PU}
  \max_i \|[\G \ \U_{\new}]^* \e_i\|^2 \le \frac{\rho_r n_{(2)}}{n_1\log^2n_{(1)}}, 
\end{equation}
\begin{equation}
  \label{eq:PV}
\max_i \|\V_\new^*\e_i\|^2 \le \frac{\rho_r n_{(2)}}{n_2\log^2n_{(1)}}, 
\end{equation}
and
\begin{equation}
  \label{eq:UV}
  \|\U_\new \V_\new^*\|_\infty \le \sqrt{\frac{\rho_r }{n_{(1)}\log^2n_{(1)}}}. 
\end{equation}

\subsection{Main Result}
We state the main result in a form that is slightly different from that of \cite{rpca}. It eliminates the parameter $\mu$ and combines the bound on $\mu r$ directly with the incoherence assumptions ($\mu$ is a parameter defined in \cite{rpca} to quantify the denseness of $\U$ and $\V$ and the incoherence between their rows) . We state it this way because it is easier to interpret and compare with the result of PCP. In particular, the dependence of the result on $n_{(2)}$ is clearer this way. The corresponding result for PCP in the same form is an immediate corollary.

\begin{theorem}
Consider the problem of recovering $\Lb$ and $\Sb$ from $\M$ using partial subspace knowledge $\G$ by solving modified-PCP (\ref{eq:sdp}). 
Assume that $\Omega$, the support set of $\Sb$, is uniformly distributed with size $m$ satisfying
\begin{equation}
\label{rc}
m \le 0.4\rho_s n_{1} n_{2}  
\end{equation}

Assume that $\Lb$ satisfies (\ref{eq:PU}), (\ref{eq:PV}) and (\ref{eq:UV}) and $\rho_s, \rho_r$, are small enough and $n_1, n_2$ are large enough to satisfy Assumption \ref{rhosr} given below.
Then, Modified-PCP  (\ref{eq:sdp}) with $\lambda=1/\sqrt{n_{(1)}}$ recovers $\Sb$ and $\Lb$ exactly with probability at least $1-23n_{(1)}^{-10}$.
\label{maintheorem}
\end{theorem}

\begin{assumption}
Assume that $\rho_s, \rho_r$ and $n_1, n_2$ satisfy:
\begin{enumerate}
\item[(a)] $\rho_r \leq \min\{10^{-4},7.2483\times10^{-5}C_{03}^{-4}\}$
\item[(b)] $\rho_s =  \min\{1-1.5b_1(\rho_r),0.0156\}$  where $b_1(\rho_r):=\max\left\{{60\rho_r^{1/2}}, 11C_{01}\rho_r^{1/2}, 0.11\right\}$ 
\item[(c)] $n_{(1)}\geq \max\left\{\exp(0.5019\rho_r), \exp(253.9618C_{01}\rho_r), 1024\right\}$                                             
\item[(d)] $n_{(2)}\geq 100\log^2n_{(1)},$
\item[(e)] $\frac{(n_1+n_2)^{1/6}}{\log (n_1+n_2)}>\frac{10.5}{(\rho_s)^{1/6}(1-5.6561\sqrt{\rho_s})},$
\item[(f)] $\frac{n_{(1)}n_{(2)}}{500\log n_{(1)}} > 1/\rho_s^2$                                      
\end{enumerate}
where $C_{01},C_{03}$ are numerical constants from Lemma \ref{POPTc0} (\cite[Theorem 4.1]{candes2009exact}) and Lemma \ref{teo:sixthree} (\cite[Theorem 6.3]{candes2009exact}) respectively. Their expressions were not specified in the original paper.
\label{assrhosrhors}
\label{rhosr}
\end{assumption}

{\em Proof: } We prove this result in Sec \ref{proof_maintheorem}. 

\subsection{Discussion w.r.t. PCP}

The PCP program of \cite{rpca} is (\ref{eq:sdp}) with no subspace knowledge available, i.e. $\G_{PCP}=[ \ ]$ (empty matrix). With this, Theorem \ref{maintheorem} simplifies to the corresponding result for PCP. Thus, $\U_{\new,PCP} = \U $ and $\V_{\new,PCP} = \V$ and so PCP needs
\begin{equation}
\max_i \|\U^*\e_i\|^2 \le   \frac{\rho_r n_{(2)}}{n_1 \log^2 n_{(1)}},
\label{eq:PU2}
\end{equation}
\begin{equation}
\max_i \|\V^*\e_i\|^2 \le   \frac{\rho_r n_{(2)}}{n_2 \log^2 n_{(1)}},
\label{eq:PV2}
\end{equation}
and
\begin{equation}
\|\U \V^*\|_\infty \le  \sqrt{\frac{\rho_r}{n_{(1)}\log^2n_{(1)}}}.
\label{eq:UV2}
\end{equation}

Notice that the second and third conditions needed by modified-PCP, i.e. (\ref{eq:PV}) and (\ref{eq:UV}), are always weaker than (\ref{eq:PV2}) and (\ref{eq:UV2}) respectively. They are much weaker when $r_\new$ is small compared to $r$. When $r_\extra = 0$, $\Span(\G) = \Span(\U_0)$ and so the first condition is the same for both modified-PCP and PCP. When $r_\extra > 0$ but is small, the first condition for modified-PCP is slightly stronger. However, as we argue below the third condition is the hardest to satisfy and hence in all cases except when $r_\extra$ is very large, the modified-PCP requirements are weaker. We demonstrate this via simulations and for some real data in Sec \ref{simu_simulated_data} (see Fig \ref{rhor_rextra2} and Fig \ref{rhor_varc}) and \ref{real_online}. 

The third condition constrains the inner product between the rows of two basis matrices $\U$ and $\V$ while the first and second conditions only constrain the norm of the rows of a basis matrix. On first glance it may seem that the third condition is implied by the first two using the Cauchy-Schwartz inequality. However that is not the case. Using Cauchy-Schwartz inequality, the first two conditions only imply that $\|\U \V^*\|_\infty \le  \sqrt{\frac{\rho_r}{n_{(1)}\log^2n_{(1)}}} \frac{\sqrt{\rho_r n_{(2)}} }{\log n_{(1)} }$ which is looser than what the third condition requires.

\section{Online robust PCA}
\label{online_robust_pca}

Consider the online / recursive robust PCA problem where data vectors $\y_t: = \s_t + \lb_t$ come in sequentially and their subspace can change over time.
 Starting with an initial knowledge of the subspace, the goal is to estimate the subspace spanned by $\lb_1, \lb_2, \dots \lb_t$ and to recover the $\s_t$'s. Assume the following subspace change model introduced in \cite{rrpcp_perf}: $\lb_t = \Pbf_{(t)} \ab_t$ where $\Pbf_{(t)} = \Pbf_j$ for all $t_j \le t < t_{j+1}$, $j=0,1,\dots J$. At the change times, $\Pbf_j$ changes as $\Pbf_j = [(\Pbf_{j-1} \Rb_j \setminus \Pbf_{j,\old}) \ \Pbf_{j,\new}]$ where $\Pbf_{j,\new}$ is a $n \times c_{j,\new}$ basis matrix that satisfies $\Pbf_{j,\new}^*\Pbf_{j-1}=0$; $\Rb_j$ is a rotation matrix; and $\Pbf_{j,\text{old}}$ is a $n \times c_{j,\text{old}}$ matrix that contains a subset of columns of $\Pbf_{j-1} \Rb_j$. Also assume that $c_{j,\new} \le c$ and $\sum_j (c_{j,\new} - c_{j,\old}) \le c_{dif}$.
 Let $r_j := \rank(\Pbf_j)$.  Clearly, $r_j = r_{j-1}+c_{j,\new} - c_{j,\old}$ and so $r_j \le r_{\max} = r_0+c_{dif}$.

For the above model, the following is an easy corollary.
\begin{corollary}[modified-PCP for online robust PCA] \label{cor}
Let $\M_j:= [\y_{t_j}, \y_{t_j + 1}, \dots \y_{t_{j+1} -1}]$, $\Lb_j:=[\lb_{t_j}, \lb_{t_j + 1}, \dots \lb_{t_{j+1} -1}]$, $\Sb_j:=[\s_{t_j}, \s_{t_j + 1}, \dots \s_{t_{j+1} -1}]$ and let $\Lb_{\text{full}}:=[\Lb_1,\Lb_2, \dots \Lb_J]$ and $\Sb_{\text{full}}:=[\Sb_1,\Sb_2, \dots \Sb_J]$.
Suppose that the following hold.
\begin{enumerate}
\item $\Sb_{\text{full}}$ satisfies the assumptions of Theorem \ref{maintheorem}.

\item The initial subspace $\Span(\Pbf_0)$ is exactly known, i.e. we are given $\hat{\Pbf}_0$ with $\Span(\hat{\Pbf}_0) = \Span(\Pbf_0)$.

\item For all $j=1,2,\dots J$, (\ref{eq:PU}),  (\ref{eq:PV}), and (\ref{eq:UV}) hold with $n_1 = n$, $n_2= t_{j+1}-t_j$, $\G = \Pbf_{j-1}$, $\U_\new = \Pbf_{j,\new}$ and $\V_\new$ being the matrix of right singular vectors of $\Lb_\new = (\Ib- \Pbf_{j-1} \Pbf_{j-1}^*) \Lb_j$. 

\item We solve modified-PCP at every $t=t_{j+1}$, using $\M=\M_j$ and with $\G = \G_j = \hat{\Pbf}_{j-1}$ where $\hat{\Pbf}_{j-1}$ is the matrix of left singular vectors of the reduced SVD of $\hat{\Lb}_{j-1}$ (the low-rank matrix obtained from modified-PCP on $\M_{j-1}$). At $t=t_1$ we use $\G=\hat{\Pbf}_0$.
\end{enumerate}
Then, modified-PCP recovers $\Sb_{\text{full}},\Lb_{\text{full}}$ exactly and in a piecewise batch fashion with probability at least $(1 - 23 n^{-10})^J$.
\end{corollary}

\begin{proof}  
Denote by $\Theta_0$ the event that $\Span(\hat{\Pbf}_0) = \Span(\Pbf_0)$. For $j=1,2,\dots J$, denote by $\Theta_j$  the event that the program (\ref{eq:sdp}) succeeds for the matrix $\M = \M_j$, i.e. $\Sb_j$ and $\Lb_j$ are exactly recovered. 
Clearly, $\Theta_j$ also implies that $\Span(\hat{\Pbf}_j) = \Span(\Pbf_j)$. 
Using Theorem \ref{maintheorem} and the model, we then get that probability $\mathbb{P}(\Theta_j| \Theta_0, \Theta_1, \dots \Theta_{j-1}) \geq  1 - 23 n^{-10}$. Also, by assumption, $\mathbb{P}(\Theta_0)=1$.
Thus by chain rule, $\mathbb{P}(\Theta_0, \Theta_1, \Theta_2, \cdots, \Theta_J) \geq (1 - 23 n^{-10})^J$.
%
\end{proof}

{\em Discussion w.r.t. PCP. } For the data model above, two possible corollaries for PCP can be stated.
\begin{corollary}[PCP for online robust PCA]
If $\Sb_{\text{full}}$ satisfies the assumptions of Theorem \ref{maintheorem} and if (\ref{eq:PU}), (\ref{eq:PV}), and (\ref{eq:UV}) hold with $n_1 = n$, $n_2= t_{J+1}-t_1$, $\G_{PCP}=[ \ ]$, $\U_{\new,PCP} = \U = [\Pbf_0, \Pbf_{1,\new}, \dots \Pbf_{J,\new}]$ and $\V_{\new,PCP}=\V$ being the right singular vectors of $\Lb_{\text{full}}:=[\Lb_1,\Lb_2, \dots \Lb_J]$, then, we can recover $\Lb_{\text{full}}$ and $\Sb_{\text{full}}$ exactly with probability at least $(1 - 23 n^{-10})$ by solving PCP (\ref{eq:pcp}) with input $\M_{\text{full}}$. Here $\M_{\text{full}}:=\Lb_{\text{full}}+\Sb_{\text{full}}$.
\end{corollary}

When we compare this with the result for modified-PCP, the second and third condition are even more significantly weaker than those for PCP. The reason is that $\V_\new$ contains at most $c$ columns while $\V$ contains at most $r_0 + Jc$ columns. The first conditions cannot be easily compared. The LHS contains at most $r_{\max} + c = r_0 + c_{dif} + c$ columns for modified-PCP, while it contains $r_0+Jc$ columns for PCP. However, the RHS for PCP is also larger. If $t_{j+1}-t_j=d$, then the RHS is also $J$ times larger for PCP than for modified-PCP. The above advantage for mod-PCP comes with two caveats. First, modified-PCP assumes knowledge of the subspace change times while PCP does not need this. Secondly, modified-PCP succeeds w.p. $(1 - 23 n^{-10})^J \ge 1-23 J n^{-10}$ while PCP succeeds w.p. $1 - 23 n^{-10}$.

Alternatively if PCP is solved at every $t=t_{j+1}$ using $\M_j$, we get the following corollary
\begin{corollary}[PCP for $\M_j$]
Solve PCP, i.e. (\ref{eq:pcp}), at $t=t_{j+1}$ using $M_j$. If $\Sb_{\text{full}}$ satisfies the assumptions of Theorem \ref{maintheorem} and if (\ref{eq:PU}),  (\ref{eq:PV}), and (\ref{eq:UV}) hold with $n_1 = n$, $n_2= t_{j+1}-t_j$, $\G_{PCP}=[ \ ]$, $\U_{\new,PCP} = \Pbf_j$ and $\V_{\new,PCP}=\V_j$ being the right singular vectors of $\Lb_j$ for all $j=1,2, \dots, J$, then, we can recover $\Lb_{\text{full}}$ and $\Sb_{\text{full}}$ exactly with probability at least $(1 - 23 n^{-10})^J$.
\end{corollary}
When we compare this with modified-PCP, the second and third condition are significantly weaker than those for PCP when $c_{j,\new}\ll r_j$. The first condition is exactly the same when $c_{j,\old} = 0$ and is only slightly stronger as long as $c_{j,\old} \ll r_j$.
{\em Discussion w.r.t. ReProCS. }
In \cite{rrpcp_allerton,rrpcp_tsp,rrpcp_perf}, Qiu et al studied the online / recursive robust PCA problem and proposed a novel recursive algorithm called ReProCS. With the subspace change model described above, they also needed the following ``slow subspace change" assumption:  $\|P_{j,\new}^*\lb_t\|$ is small for sometime after $t_j$ and increases gradually. Modified-PCP does not need this. Moreover, even with perfect initial subspace knowledge, ReProCS cannot achieve exact recovery of $\s_t$ or $\lb_t$ while, as shown above, modified-PCP can. On the other hand, ReProCS is a recursive algorithm while modified-PCP is not; and for highly correlated support changes of the $\s_t$'s, ReProCS outperforms modified-PCP (see Sec \ref{simu}). The reason is that correlated support change results in $\Sb$ also being rank deficient, thus making it difficult to separate it from $\Lb_\new$ by modified-PCP. 


{\em Discussion w.r.t. the work of Feng et al. }
Recent work of Feng et. al. \cite{xu_nips2013_1,xu_nips2013_2} provides two asymptotic results for online robust PCA. The first work \cite{xu_nips2013_1} does not model the outlier as a sparse vector but just as a vector that is ``far" from the low-dimensional data subspace. 
In \cite{xu_nips2013_2}, the authors reformulate the PCP program and use this to develop a recursive algorithm that comes ``close" to the PCP solution asymptotically.





\section{Proof of Theorem \ref{maintheorem}: main lemmas}
\label{proof_maintheorem}
Our proof adapts the proof approach of \cite{rpca} to our new problem and the modified-PCP solution. The main new lemma is Lemma \ref{kktrc} in which we obtain different and weaker conditions on the dual certificate to ensure exact recovery. This lemma is given and proved in Sec \ref{dualcerts}. In addition, we provide a proof for two key statements from \cite{rpca} for which either a proof is not immediate (Lemma \ref{unifber}) or for which the cited reference does not work (Lemma \ref{sign2norm}). These lemmas are given below in Sec \ref{new_lems} and proved in the Appendix.

We state Lemma \ref{unifber} and Lemma \ref{sign2norm} in Sec \ref{new_lems}. We give the overall proof architecture next in Sec \ref{outline}. Some definitions and basic facts are given in Sec \ref{defs} and \ref{facts}. In Sec \ref{dualcerts}, we obtain sufficient conditions (on the dual certificate) under which $\Sb,\Lb_\new$ is the unique minimizer of modified-PCP. In Sec \ref{golfing}, we construct a dual certificate that satisfies the required conditions with high probability (w.h.p.). Here, we also give the two main lemmas to show that this indeed satisfies the required conditions. The proof of all the four lemmas from this section is given in the Appendix.

Whenever we say ``with high probability" or w.h.p., we mean with probability at least $1 - O(1)n_{(1)}^{-10}$.

\subsection{Two Lemmas} \label{new_lems}
\begin{lemma} 
Denote by $\Pb_{\Unif}$ and $\Pb_{\Ber}$ the probabilities calculated under the uniform and Bernoulli models and let ``Success" be the event that $(\Lb_\new,\Sb,\Lb^*\G)$ is the unique solution of modified-PCP (\ref{eq:sdp}). Then
$$\Pb_{\Unif(m_0)}(\Succ) \geq \Pb_{\Ber(\rho_0)}(\Succ) - e^{-2n_1n_2\epsilon_0^2},$$
where $\rho_0=\frac{m_0}{n_1n_2}+\epsilon_0$.
\label{unifber}
\end{lemma}

The proof is given in Appendix \ref{proofunifber}. A similar statement is given in Appendix A.1 of \cite{rpca} but without a proof.
The expression for the second term on the right hand side given there is $e^{-\frac{2n_1n_2\epsilon_0^2}{\rho_0}}$ which is different from the one we derive.
\\


\begin{lemma}
Let $\Eb$ be a $n_1\times n_2$ random matrix with entries i.i.d. (independently identically distributed) as
\begin{equation}
\label{eq:random-sign}
\Eb_{ij} = \begin{cases} 1, & \text{w.~p. }
  \rho_s/2,\\
  0, & \text{w.~p. } 1-\rho_s,\\
  -1, & \text{w.~p. } \rho_s/2.
\end{cases}
\end{equation}
If $\rho_s < 0.03$ and $\frac{(n_1+n_2)^{1/6}}{\log (n_1+n_2)}>\frac{10.5}{(\rho_s)^{1/6}(1-5.6561\sqrt{\rho_s})}$, then  
$$\Pb(\|\Eb\|\geq 0.5\sqrt{n_{(1)}}) \leq n_{(1)}^{-10}.$$
\label{sign2norm}
\end{lemma}

The proof is provided in Appendix \ref{proofsign2norm} and uses the result of \cite{furedi1981eigenvalues}. In \cite{rpca}, the authors claim that using \cite{vershynin2010introduction}, $\|\Eb\| > 0.25 \sqrt{n_{(1)}}$ w.p. less than  $n_{(1)}^{-10}$. While the claim is correct, it is not possible to prove it using any of the results from \cite{vershynin2010introduction}.
Using ideas from \cite{vershynin2010introduction}, one can only show that the above holds when $n_{(2)}$ is upper bounded by a constant times $\log n_{(1)}$ (see the Appendix \ref{vershynin_bound}) which is a strong extra assumption. 


\subsection{Proof Architecture} \label{outline}

The proof of the theorem involves 4 main steps.
\begin{itemize}
\item[(a)] The first step is to show that when the locations of the support of $\Sb$ are Bernoulli distributed with parameter $\rho_s$ and the signs of $\Sb$ are i.i.d $\pm1$ with probability $1/2$ (and independent from the locations), and all the other assumptions on $\Lb, n_1, n_2, \rho_s, \rho_r$ in Theorem \ref{maintheorem} are satisfied, then Modified-PCP  (\ref{eq:sdp}) with $\lambda=1/\sqrt{n_{(1)}}$ recovers $\Sb$ exactly (and hence also $\Lb=\M-\Sb$) with probability at least $1-22n_{(1)}^{-10}$.

\item[(b)] By \cite[Theorem 2.3]{rpca}, the previous claim also holds for the model in which the signs of $\Sb$ are fixed and the locations of its nonzero entries are sampled from the Bernoulli model with parameter $\rho_s/2$, and all the other assumptions on $\Lb, n_1, n_2, \rho_s, \rho_r$ from Theorem \ref{maintheorem} are satisfied.

\item[(c)] By Lemma \ref{unifber} with $\epsilon_0=0.1\rho_s$, $m_0 = \lfloor0.4\rho_sn_1n_2\rfloor$, since $n_1n_2 > 500 \log n_1 / \rho_s^2$ (Assumption \ref{assrhosrhors}(f)), the previous claim holds with probability at least $1-23n_{(1)}^{-10}$ for the model in which the signs of $\Sb$ are fixed and the locations of its nonzero entries are sampled from the Uniform model with parameter $m_0$, and all the other assumptions on $\Lb, n_1, n_2, \rho_s, \rho_r$ from Theorem \ref{maintheorem} are satisfied.

\item[(d)] By \cite[Theorem 2.2]{rpca}, the previous claim also holds for the model in which the signs of $\Sb$ are fixed and the locations of its nonzero entries are sampled from the Uniform model with parameter $m \leq m_0 = 0.4\rho_sn_1n_2$, and all the other assumptions on $\Lb, n_1, n_2, \rho_s, \rho_r$ from Theorem \ref{maintheorem} are satisfied.
\end{itemize}

Thus, all we need to do is to prove step (a).
To do this we start with the KKT conditions and strengthen them to get a set of easy to satisfy sufficient conditions on the dual certificate under which $\Lb_\new, \Sb$ is the unique minimizer of (\ref{eq:sdp}). This is done in Sec \ref{dualcerts}. Next, we use the golfing scheme \cite{gross2010quantum,rpca} to construct a dual certificate that satisfies the required conditions (Sec. \ref{golfing}). 

\subsection{Basic Facts} \label{facts}
We state some basic facts which will be used in the following proof.
%
\begin{definition}[Sub-gradient \cite{subgradient}]
Consider a convex function $f:\mathbb{O}\rightarrow \mathbb{R}$ on a convex set of matrices $\mathbb{O}$. A matrix $\Y$ is called its sub-gradient at a point $\X_0 \in \mathbb{O}$ if
$$f(\X)-f(\X_0) \geq \<\Y,(\X-\X_0)\>.$$
for all $\X \in \mathbb{O}$. The set of all sub-gradients of $f$ at $\X_0$ is denoted by $\partial f(\X_0)$.
%
\end{definition}

It is known \cite{lewis2003mathematics,watson1992characterization} that
$$\partial \|\Lb_\new\|_* = \{\U_\new \V_\new^*+\W: \PT \W=0, \|\W\| \leq 1\}.$$
and
$$\partial \|\Sb\|_1 = \{\F: \PO \F= \sgn(\Sb), \|\F\|_\infty \leq 1\}.$$

\begin{definition}[Dual norm \cite{RMD}]
The matrix norm $\|\cdot\|_\heartsuit$ is said to be dual to matrix norm $\|\cdot\|_\spadesuit$ if, for all $\Y_1\in \mathbb{R}^{n_1\times n_2}$, $\|\Y_1\|_\heartsuit = \sup_{\|\Y_2\|_\spadesuit \leq 1} \langle \Y_1, \Y_2\rangle$.
\end{definition}

\begin{proposition}[Proposition 2.1 of \cite{recht2010guaranteed}]
The following pairs of matrix norms are dual to each other:
\begin{itemize}
\item $\|\cdot\|_1$ and $\|\cdot\|_\infty$;
\item $\|\cdot\|_*$ and $\|\cdot\|$;
\item $\|\cdot\|_F$ and $\|\cdot\|_F.$
\end{itemize}
For all these pairs, the following hold.
\begin{enumerate}
\item $ | \<\Y,\Z\> | \le \|\Y\|_\spadesuit\|\Z\|_\heartsuit.$
\item Fixing any $\Y\in \mathbb{R}^{n_1\times n_2}$, there exists $\Z\in \mathbb{R}^{n_1\times n_2}$ (that depends on $\Y$) such that
$$\<\Y,\Z\>=\|\Y\|_\spadesuit\|\Z\|_\heartsuit.$$
\item In particular, we can get $\<\Y,\Z\>=\|\Y\|_1 \|\Z\|_\infty$ by setting $\Z = \sgn(\Y)$, 
we can get $\<\Y,\Z\>=\|\Y\|_*\|\Z\|$ by setting $\Z=\U_Y\V_Y^*$ where $\U_Y\Sigmab_Y\V_Y^*$ is the SVD of $\Y$, and we can get  $\<\Y,\Z\>=\|\Y\|_F \|\Z\|_F$ by letting $\Z = \Y$.
\end{enumerate}
\label{dual_norm_lem}
\end{proposition}


For any matrix $\Y$, we have
\begin{align*}
\|\Y\|_F^2 = \trace(\Y^*\Y) = \sum_{i,j}|\Y_{ij}|^2 \leq (\sum_{i,j}|\Y_{ij}|)^2 = \|\Y\|_1^2
\end{align*}
and
\begin{align*}
\|\Y\|_F^2 = \trace(\Y^*\Y)= \sum_{i} \sigma_i^2(\Y) \leq (\sum_{i} \sigma_i(\Y))^2 = \|\Y\|_*^2 
\end{align*}

Let $\Upsilon$ be the linear space of matrices with column span equal to that of the columns of $\Pbf_1$ and row span equal to that of the columns of $\Pbf_2$ where $\Pbf_1$ and $\Pbf_2$ are basis matrices. 
Then, for a matrix $\M$,
$${\mathcal{P}_{\Upsilon^\perp}} \M = (\Ib- \Pbf_1 \Pbf_1^*) \M (\Ib- \Pbf_2 \Pbf_2^*)  \ \text{and} \ {\mathcal{P}_{\Upsilon}} \M = \M - {\mathcal{P}_{\Upsilon^\perp}} \M.$$
Let $\Upsilon$ be the linear space of matrices with column span equal to that of the columns of $\Pbf_1$. 
Then,
$${\mathcal{P}_{\Upsilon^\perp}} \M = (\Ib- \Pbf_1 \Pbf_1^*) \M \ \text{and} \ {\mathcal{P}_{\Upsilon}} \M = \Pbf_1 \Pbf_1^* \M $$

For a matrix $\x\y^*$ where $\x$ and $\y$ are vectors,
$$\|\x\y^*\|_F^2 = \|\x\|^2 \|\y\|^2.$$

If an operator $\mathcal{A}$ is linear and bounded, then \cite{operatornorm}
$$\|\mathcal{A}^*\mathcal{A}\|=\|\mathcal{A}\|^2.$$

\subsection{Definitions} \label{defs}
Here we define the following linear spaces of matrices.

Denote by $\Gamma$ the linear space of matrices with column span equal to that of the columns of $\G$, i.e. 
\begin{equation}
\label{eq:G}
\Gamma := \{\G \Y^*, \, \Y \in \mathbb{R}^{n_2 \times r_G}\},
\end{equation}
and by $\Gamma^{\perp}$ its orthogonal complement.

Define also the following linear spaces of matrices
\begin{equation}\label{eq:T}
T_\new := \{\U_\new \Y_1^* + \Y_2\V_\new^*, \, \Y_1\in \mathbb{R}^{n_2 \times r_\new}, \Y_2 \in \mathbb{R}^{n_1 \times r_\new}\},\nonumber
\end{equation}
\begin{equation}\label{eq:Pi}
\Pi := \{[\G \ \U_\new]\Y_1^* + \Y_2\V_\new^*, \, \Y_1\in \mathbb{R}^{n_2 \times (r_G+r_\new)}, \Y_2 \in \mathbb{R}^{n_1 \times r_\new}\},\nonumber
\end{equation}
Notice that $T_\new \cup \Gamma = \Pi.$

\begin{remark}
For the matrix $\e_i \e_j^*$, together with (\ref{eq:PU}) and (\ref{eq:PV}), we have
\begin{equation}
\begin{array}{ll}
&\|\PPp \e_i \e_j^*\|_F^2\\=& \|(\Ib- [\G \ \U_\new][\G \ \U_\new]^*) \e_i\|^2 \|(\Ib- \V_\new \V_\new ^*) \e_j\|^2 \\
\ge& (1-\rho_r/\log^2n_{(1)})^2,
\end{array}
\end{equation}
where $\rho_r/\log^2n_{(1)} \le 1$ as assumed. Using $\|\PP \e_i\e_j^*\|_F^2 +
\|\PPp \e_i \e_j^*\|_F^2 = 1$, we have
\begin{equation}
  \label{eq:PPeiej}
  \|\PP \e_i \e_j^*\|_F \le \sqrt{\frac{2\rho_r}{\log^2n_{(1)}}}.
\end{equation}
\end{remark}



\subsection{Dual Certificates} \label{dualcerts}

We modify Lemma 2.5 of \cite{rpca} to get the following lemma which gives us sufficient conditions on the dual certificate needed to ensure that modified-PCP succeeds.

\begin{lemma} 
\label{kktrc}
If $\|\PO\PP\| \le 1/4$, $\lambda < 3/10$, and there is a pair $(\W,\F)$ obeying
\[
\U_\new \V_\new^* +  \W = \lambda(\sgn(\Sb) + \F + \PO \D)
\]
with $\PP \W = \zero$, $\|\W\| \le \frac{9}{10}$, $\PO \F = \zero$,  $\|\F\|_\infty \le \frac{9}{10}$, and $\|\PO \D\|_F \le \frac{1}{4}$, then
$(\Lb_\new, \Sb, \Lb^*\G)$ is the unique solution to Modified-PCP (\ref{eq:sdp}).
\end{lemma}

\begin{proof}
 Any feasible perturbation of $(\Lb_\new,\Sb,\Lb^* \G)$ will be of the form
 $$(\Lb_\new + \Hb_1, \Sb-\Hb, \Lb^*\G + \Hb_2), \ \text{with} \  \Hb_1+\G \Hb_2^*=\Hb.$$
Let $\G_\perp$ be a basis matrix that is such that $[\G \ \G_\perp]$ is a unitary matrix.
Then, $\Hb_1=\Hb-\G\Hb_2^* = \G_\perp \G_\perp^* \Hb + \G\G^*\Hb- \G \Hb_2^* $.
Notice that
\begin{itemize}
\item  $\Lb_\new = \G_\perp \G_\perp^*\Lb_\new$ and $ \G_\perp^*  \G_\perp^* \Hb = \PGp \Hb$.
\item For any two matrices $\Y_1$ and $\Y_2$,
$$\| \G_\perp \Y_1+ \G \Y_2\|_* \ge \| \G_\perp \Y_1\|_*$$
where equality holds if and only if $\Y_2 = \zero$.
To see why this holds, let the full SVD of $\Y_1,\Y_2$ be $\Y_1 \stackrel{\text{SVD}}{=} \Q_1\Sigmab_1\V_1^*$ and $\Y_2 \stackrel{\text{SVD}}{=} \Q_2\Sigmab_2\V_2^*$. Since $[\G \ \G_\perp]$ is a unitary matrix,  $\G_\perp \Y_1+ \G \Y_2 \stackrel{\text{SVD}}{=} [\G_\perp \Q_1\ \G \Q_2]\left[\Sigmab_1\ \zero\atop \zero\ \Sigmab_2\right][\V_1\ \V_2]^*$. Thus,
$\| \G_\perp \Y_1+ \G \Y_2\|_*= \text{trace}(\Sigmab_1)+\text{trace}(\Sigmab_2) \geq \text{trace}(\Sigmab_1)=\|\G_\perp \Y_1\|_*$ where equality holds if and only if $\Sigmab_2=\zero$, or equivalently, $\Y_2=\zero$.
\end{itemize}
Thus,
\begin{eqnarray}
&& \|\Lb_\new + \Hb_1\|_* \nonumber \\
&&= \|\G_\perp(\G_\perp^*\Lb_\new + \G_\perp^*\Hb) + \G(\G^*\Hb-\Hb_2^*)\|_* \nonumber \\
&& \geq  \|\G_\perp(\G_\perp^*\Lb_\new + \G_\perp^* \Hb)\|_* = \|\Lb_\new + \PGp \Hb\|_*  \ \ \ \ \
\label{eqH1}
\end{eqnarray}
where equality holds if and only if $\Hb_2 = \G^*\Hb$.

Recall that $T_\new \cup \Gamma = \Pi.$
Choose a $\W_a$ so that $\< \W_a, \PPp \Hb\> = \|\PPp \Hb\|_* \|\W_a\|$. This is possible using Proposition \ref{dual_norm_lem}. Let
$$\W_0 = \PPp \W_a/\|\W_a\|.$$
Thus, $\W_0$ satisfies $\PT \W_0=\zero$ and $\|\W_0\| \le 1$ and so it belongs to the sub-gradient set of the nuclear norm at $\Lb_\new$. Also,
\begin{eqnarray}
\< \W_0, \PGp \Hb\> &=& \frac{1}{\|\W_a\|} \< \PPp \W_a, \PGp \Hb\>  \nonumber \\
                 &=&  \frac{1}{\|\W_a\|} \< \W_a, \PPp \PGp \Hb\> \nonumber \\
                 &=& \frac{1}{\|\W_a\|} \< \W_a, \PPp \Hb\> = \|\PPp \Hb\|_*. \nonumber
\end{eqnarray}
%
Let $\F_0 = - \sgn(\POp \Hb)$. Thus, $\PO \F_0=\zero$, $\|\F_0\|_\infty=1$ and so it belongs to the sub-gradient set of the 1-norm at $\Sb$. Also,
$$\< \F_0, \Hb\>= \< \F_0, \POp \Hb\>  = - \|\POp \Hb\|_1.$$
Thus,
\begin{align*}
  &\|\Lb_\new + \Hb_1\|_* + \lambda \|\Sb - \Hb\|_1 \\
\ge  &\|\Lb_\new + \PGp \Hb\|_* + \lambda \|\Sb - \Hb\|_1 \\
&\text{(using (\ref{eqH1}))}\\
  \ge & \|\Lb_\new\|_* + \lambda\|\Sb\|_1 + \<\U_\new \V_\new^* + \W_0, \PGp \Hb\> \\
      & - \lambda \<\sgn(\Sb) + \F_0, \Hb\>\\
  &\text{(by definition of sub-gradient)}\\
   = & \|\Lb_\new\|_* + \lambda\|\Sb\|_1 + \|\PPp \Hb\|_* + \lambda \|\POp \Hb\|_1   \\
    & + \<\U_\new \V_\new^*- \lambda\sgn(\Sb), \Hb\> \\
  & \text{(using $\W_0$ and $\F_0$ as defined above)}\\
 \ge& \|\Lb_\new\|_* + \lambda\|\Sb\|_1 + \|\PPp \Hb\|_* + \lambda \|\POp \Hb\|_1  \\
   & - \text{max}(\|\W\|, \|\F\|_\infty)(\|\PPp  \Hb\|_* + \lambda \|\POp \Hb\|_1)+ \lambda \<\PO \D,  \Hb\>\\
  & \text{(by the lemma's assumption and Proposition \ref{dual_norm_lem})}\\
%
%
 \ge & \|\Lb_\new\|_* + \lambda \|\Sb\|_1 + \frac{1}{10} \Bigl(\|\PPp \Hb\|_* +  \lambda \|\POp \Hb\|_1\Bigr) \\
  & -  {\lambda \over 4} \|\PO \Hb\|_F\\
  &\text{(by  Proposition \ref{dual_norm_lem} and assumption $\|\PO \D\|_F \le \frac{1}{4}$)}
\end{align*}
Observe now that
\begin{align*}
  \|\PO \Hb\|_F & \le \|\PO \PP \Hb\|_F  + \|\PO \PPp \Hb\|_F  \\
  & \le \frac{1}{4} \|\Hb\|_F + \|\PPp \Hb\|_F \\  
  & \le \frac{1}{4}\|\PO \Hb\|_F + \frac{1}{4} \|\POp \Hb\|_F  + \|\PPp \Hb\|_F
\end{align*}
and, therefore,
\begin{align*}
 \|\PO \Hb\|_F & \le  \frac{1}{3}\|\POp \Hb\|_F  + \frac{4}{3}\|\PPp \Hb\|_F \\
     & \le \frac{1}{3}\|\POp \Hb\|_1  + \frac{4}{3}\|\PPp \Hb\|_*
\end{align*}
In conclusion,
\begin{align*}
&\|\Lb_\new + \PGp \Hb\|_* + \lambda \|\Sb - \Hb\|_1 \\
& \ge \|\Lb_\new\|_* + \lambda \|\Sb\|_1 + \frac{}{} \Bigl((\frac{1}{10}-\frac{\lambda}{3}) \|\PPp \Hb\|_* + \frac{\lambda }{60}
\|\POp \Hb\|_1\Bigr) \\
& > \|\Lb_\new\|_* + \lambda \|\Sb\|_1
\end{align*}
The last inequality holds because $\|\PO \PP\| < 1$ and this implies that $\Pi \cap \Omega = \{0\}$ and so at least one of $\PPp \Hb$ or $\POp \Hb$ is strictly positive for $\Hb \neq \zero$. Thus, the cost function is strictly increased by any feasible perturbation. Since the cost is convex, this proves the lemma.
\end{proof}

Lemma \ref{kktrc} is equivalently saying that $(\Lb_\new,\Sb,\Lb^* \G)$ is the unique solution to Modified-PCP (\ref{eq:sdp}) if there is a $\W$ satisfying:
\begin{equation}
\label{dual-certif-dg}
\begin{cases} \W \in \Pi^\perp, \\
  \|\W\| \le 9/10,\\
  \|\PO(\U_\new \V_\new^* -\lambda \sgn(\Sb) +  \W)\|_F \le \lambda/4,\\
  \|\POp(\U_\new \V_\new^* +  \W)\|_\infty < 9\lambda/10.\\
\end{cases}
\end{equation}

\subsection{Construction of the required dual certificate}
\label{golfing}
The golfing scheme is introduced by \cite{gross2011recovering,gross2010quantum}; here we use it with some modifications similar to those in \cite{rpca} to construct dual certificate. Assume that $\Omega \backsim \Ber(\rho_s)$ or equivalently, $\Omega^c \backsim \Ber(1-\rho_s)$. 

Notice that $\Omega^c$ can be generated as a union of $j_0$ i.i.d. sets $\{\bOmega_j\}_{j=1}^{j_0}$, where $\bOmega_j \stackrel{i.i.d}{\backsim} \Ber(q), 1\leq j\leq j_0$ with $q,j_0$ satisfying $\rho_s=(1-q)^{j_0}$. This is true because
$$\mathbb{P}((i,j)\in \Omega)=\mathbb{P}((i,j)\notin \bOmega_1\cup \bOmega_2\cup \cdots \bOmega_{j_0} )=(1-q)^{j_0}.$$
As there is overlap between $\bOmega_j's$, we have $q\geq (1-\rho_s)/j_0$.\\
Let $\W=\W^L+\W^S$, where $\W^L, \W^S$ are constructed similar to \cite{rpca} as:
\begin{itemize}
\item {\em Construction of $\W^L$ via the golfing scheme.} Let $\Y_0=0$, $$\Y_j=\Y_{j-1}+q^{-1}\Pc_{\bOmega_j}\PP(\U_\new \V_\new^*-\Y_{j-1}),$$ and $\W^L=\PPp \Y_{j_0}.$ Notice that $\Y_j \in \Omega^\perp$.

\item {\em Construction of $\W^S$ via the method of least squares.}
Assume that $\|\PO\PP\| \le 1/4$. We prove that this holds in Lemma \ref{WSc} below. With this, $\|\PO\PP\PO\| = \|\PO\PP\|^2 \le 1/16$ and so $\| \PO-\PO\PP\PO \| \ge 1 - 1/16 > 0$. Thus this operator, which maps the subspace $\Omega$ onto itself, is invertible. Let $(\PO-\PO\PP\PO)^{-1}$ denote its inverse and let
$$\W^S=\lambda \PPp(\PO-\PO\PP\PO)^{-1}\text{sgn}(\Sb).$$
Using the Neumann series, notice that \cite{rpca}
$$(\PO-\PO\PP\PO)^{-1}\text{sgn}(\Sb)= \sum_{k \ge 0} (\PO \PP \PO)^k \sgn(\Sb).$$
\end{itemize}
Thus \cite{rpca},
$$\PO \W^S = \lambda \sgn(\Sb).$$
This follows because $(\PO - \PO \PP \PO)$ is an operator mapping $\Omega$ onto itself, and so $(\PO - \PO \PP \PO)^{-1}\sgn(\Sb) = \PO (\PO - \PO \PP \PO)^{-1}\sgn(\Sb)$ \footnote{This is also clear from the Neumann series}. With this, $\PO \W^S = \lambda \PO( \I- \PP) \PO (\PO - \PO \PP \PO)^{-1}\sgn(\Sb) = \lambda (\PO - \PO \PP \PO) (\PO - \PO \PP \PO)^{-1}\sgn(\Sb) = \lambda \sgn(\Sb)$.


Clearly, $\W = \W^L+\W^S$ is a dual certificate if
\begin{equation}
\label{dual-certif-use}
\begin{cases}
  \|\W^L + \W^S\| < 9/10,\\
  \|\PO(\U_\new \V_\new^* + \W^L)\|_F \le \lambda/4,\\
  \|\POp(\U_\new \V_\new^* + \W^L + \W^S)\|_\infty < 9\lambda/10.\\
\end{cases}
\end{equation}

Next, we present the two lemmas that together prove that (\ref{dual-certif-use}) holds w.h.p..

\begin{lemma}
\label{WLc}
Assume $\Omega \sim \Ber(\rho_s)$. Let $j_0=1.3\lceil \log n_{(1)}\rceil$. Under the other assumptions of Theorem \ref{maintheorem}, the matrix $\W^L$ obeys, with probability at least $1-11n_{(1)}^{-10}$,
\begin{enumerate}
\item[(a)] $\|\W^L\| < 1/16$,
\item[(b)] $\|\PO(\U_\new \V_\new^* +  \W^L)\|_F < \lambda/4$,
\item[(c)] $\|\POp(\U_\new \V_\new^* + \W^L)\|_\infty < 2\lambda/5$.
\end{enumerate}
\end{lemma}
This is similar to \cite[Lemma 2.8]{rpca}. The proof is in the Appendix.

\begin{lemma}
\label{WSc}
Assume $\Omega \sim \Ber(\rho_s)$, and the signs of $\Sb$ are independent of $\Omega$ and i.i.d. symmetric. Under the other assumptions of Theorem \ref{maintheorem},  with probability at least $1-11n_{(1)}^{-10}$, the following is true
\begin{enumerate}
\item[(a)] $\|\PO\PP\| \le 1/4$ and so $\W_S$ constructed earlier is well defined.
\item[(b)] $\|\W^S\| < 67/80$,
\item[(c)] $\|\POp \W^S\|_\infty < \lambda/2$.
\end{enumerate}
\end{lemma}
This is similar to \cite[Lemma 2.9]{rpca}. The proof is in the Appendix.



\section{Solving the Modified-PCP program and experiments with it}
\label{simu}

We first give below the algorithm used to solve modified-PCP. Next, we give recovery error comparisons for static simulated and real data. Finally we show some online robust PCA experiments, both on simulated and real data.

\subsection{Algorithm for solving Modified-PCP}
We give below an algorithm based on the Inexact Augmented Lagrange Multiplier (ALM) method \cite{alm} to solve the modified-PCP program, i.e. solve (\ref{eq:sdp}). This algorithm is a direct modification of the algorithm designed to solve PCP in \cite{alm} and uses the idea of \cite{hale2008fixed,cai2010singular} for the sparse recovery step.

For the modified-PCP program (\ref{eq:sdp}), the Augmented Lagrangian function is:
\begin{equation*}
\begin{array}{ll}
\mathbb{L}(\tL_\new,\tS,\Y,\tau) = &\|\tL_\new\|_* + \lambda \|\tS\|_1 + \langle \Y,  \M - \tL_\new-\tS \\ &- \G\tX^*\rangle+
\dfrac{\tau}{2}\|\M-\tL_\new- \tS - \G\tX^*\|_F^2,
\end{array}
\end{equation*}
Thus, with similar steps in \cite{alm}, we have following algorithm.
\begin{algorithm}[th]
\caption{{\bf Algorithm for solving Modified-PCP (\ref{eq:sdp})}}
\begin{algorithmic}[1]
\REQUIRE Measurement matrix $\M \in \R^{n_1 \times n_2}$, $\lambda=1/\sqrt{\max\{n_1, n_2\}}$, $\G$.
\STATE $\Y_0= \M/\max\{\|\M\|, \|\M\|_\infty/\lambda\}$; $\tS_0=0$; $\tau_0 > 0$; $v > 1$; $k = 0$.
\WHILE{not converged} 
 \STATE $\tS_{k+1} = {\mathfrak{S}}_{\lambda\tau_k^{-1}}[\M-\G\tX_{k}-\tL_{\new,k}+\tau_k^{-1} \Y_k]$.
 \STATE $(\tilde{\U},\tilde{\Sigmab},\tilde{\V}) =\mbox{svd}((I-\G\G^*) (\M- \tS_{k+1}+\tau_k^{-1} \Y_k))$;
\STATE $\tL_{\new,k+1} =
\tilde{\U}{\mathfrak{S}}_{\tau_k^{-1}}[\tilde{\Sigmab}]\tilde{\V}^T $.
\STATE $\tX_{k+1}= \G^* (\M- \tS_{k+1}+\tau_k^{-1} \Y_k)$
\STATE $\Y_{k+1} = \Y_k + \tau_k (\M - \tS_{k+1}- \tL_{\new, k+1} - \G\tX_{k+1})$.
\STATE $\tau_{k+1}=\min(v \tau_{k}, \bar{\tau})$. \STATE $k \gets k+1$.
\ENDWHILE \ENSURE $\hat{\Lb}_\new=\tL_{\new,k}, \hat{\Sb}=\tS_k, \hat{L}=\M-\tS_k$.
\end{algorithmic}
\label{inexact_alm}
\end{algorithm}
In Algorithm \ref{inexact_alm}, Lines 3 solves $\tS_{k+1} =\arg\min\limits_{\tS} \|\tL_{\new,k}\|_* + \lambda \|\tS\|_1 + \langle \Y_k,  \M - \tL_{\new,k}-\tS - \G\tX_k^*\rangle+
\dfrac{\tau}{2}\|\M-\tL_{\new,k}- \tS - \G\tX_k^*\|_F^2$; Line 4-6 solve $[\tL_{\new, k+1}, \tX_{k+1}] =\arg\min\limits_{\tL_\new, \tX} \|\tL_{\new}\|_* + \lambda \|\tS_{k+1}\|_1 + \langle \Y_k,  \M - \tL_{\new}-\tS_{k+1} - \G\tX^*\rangle+
\dfrac{\tau}{2}\|\M-\tL_{\new}- \tS_{k+1} - \G\tX_k^*\|_F^2$. The soft-thresholding operator is defined as
\begin{equation}
\mathfrak{S}_{\epsilon}[x]=\left\{\begin{array}{ll} x-\epsilon,&\text{ if }x>\epsilon;\\x+\epsilon,&\text{ if }x<-\epsilon;\\ 0, & \text{ otherwise,} \end{array}\right.
\end{equation}
Parameters are set as suggested in \cite{alm}, i.e., $\tau_0=1.25/\|\M\|, v=1.5, \bar{\tau}=10^7\tau_0$ and iteration is stopped when $\|\M - \tS_{k+1}- \tL_{\new, k+1} - \G\tX_{k+1}\|_F/\|\M\|_F<10^{-7}$.

\subsection{Simulated data}
\label{simu_simulated_data}

The data was generated as follows. For the sparse matrix $\Sb$, we generated a support set of size $m$ uniformly at random and assigned values $\pm1$ with equal probability to entries in the support set.
We generated the matrix $[\G \ \U_\new]$ by orthonormalizing an $n_1\times (r_0+ r_\extra + r_\new)$ matrix with entries i.i.d. Gaussian $\mathcal{N}(0,1/n_1)$; we set $\U_0$ as the first $r_0$ columns of this matrix, $\G_\extra$ as the next $r_\extra$ columns and $\U_\new$ as the last $r_\new$ columns. Then, we set $\G = [\U_0, \ \G_\extra]$. This matrix has $r_G = r_0 + r_\extra$ columns. We generated a matrix $\Y_1$ of size $r_G\times d$ and a matrix $\Y_2$ of size $(r_0 + r_\new)\times n_2$ with entries i.i.d. $\mathcal{N}(0,1/n_1)$.
We set $\M_G=\G \Y_1$ as training data and $\M =[\U_0 \ \U_\new]\Y_2 + \Sb$. The matrix $\M_G$ is $n_1 \times d$ and the $\M$ is $n_1 \times n_2$. We computed $\G$ as the left singular vectors with nonzero singular values of $\M_G$ and this was used as the partial subspace knowledge for modified-PCP.

For modified-PCP, we solved (\ref{eq:sdp}) with $\M$ and $\G$ using Algorithm \ref{inexact_alm}. For PCP, we solved (\ref{eq:pcp}) with $\M$ using the Inexact Augmented Lagrangian Multiplier algorithm from \cite{alm}. This section provides a simulation comparison of what we conclude from the theoretical results. In the theorems, both modified-PCP and PCP use the same matrix $\M$, but modified-PCP is given extra information (partial subspace knowledge). In the first set of simulations, we also compare with PCP when it is also given access to the initial data $\M_G$, i.e. we also solve PCP using $[\M_G\ \M]$. We refer to this as PCP($[\M_G\ \M]$).

Sparse recovery error is calculated as $\|\Sb-\hat{\Sb}\|_F^2/\|\Sb\|_F^2$ averaged over 100 Monte Carlo trials. For the simulated data, we also compute the smallest value of $\rho_r$ required to satisfy the sufficient conditions -- (\ref{eq:PU}), (\ref{eq:PV}), (\ref{eq:UV}) for mod-PCP and (\ref{eq:PU2}), (\ref{eq:PV2}), (\ref{eq:UV2}) for PCP. We denote the respective values of $\rho_r$ by $\rho_r([\G\ \U_\new])$, $\rho_r(\V_\new)$, $\rho_r(\U_\new \V_\new)$, $\rho_r(\U)$, $\rho_r(\V)$ and $\rho_r(\U\V)$.
Also,
$$\rho_r(\text{mod-PCP})=\max\{\rho_r([\G\ \U_\new]), \rho_r(\V_\new), \rho_r(\U_\new \V_\new)\}$$ and
$$\rho_r(\text{PCP})=\max\{\rho_r(\U), \rho_r(\V), \rho_r(\U\V)\}.$$ 

In Fig. \ref{rextra2}, we show comparisons with increasing number of extra directions $r_\extra$. We used $n_1=200$, $d=200$, $n_2=120$, $m=0.075n_1n_2$, $r=20$, $r_0=0.9r=18$, $r_\new=0.1r=2$ and $r_\extra$ ranging from $0$ to $n_2-r=100$. As we can see from Fig. \ref{spar_rextra2}, for $r_\extra < 60$, mod-PCP performs better than PCP with or without training data $\M_G$.
Fig. \ref{rhor_rextra2} shows that mod-PCP allows a larger value of $\rho_r$ (needs weaker assumptions) than PCP. 
Notice that the recovery error of PCP($[\M_G\ \M]$) is larger than that of PCP($\M$). This is because the rank of $[\M_G \ \M]$ is larger than that of $\M$ because of the extra directions. In the rest of the simulations, we only compare with PCP($\M$).

\begin{figure}[h!]
\centering
\begin{subfigure}{0.4\textwidth}
\includegraphics[width=\textwidth]{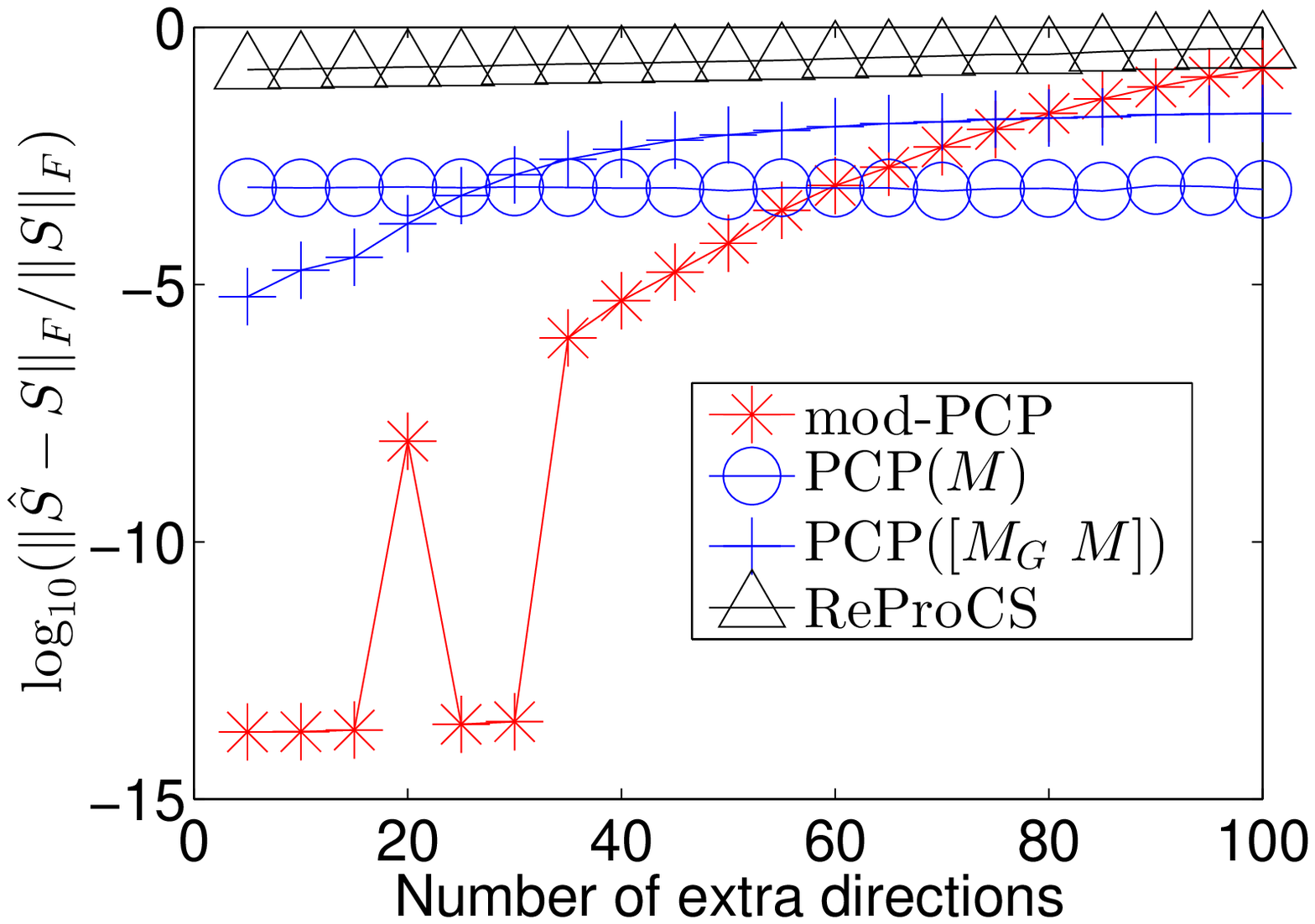}
\caption{Recovery result comparison}
\label{spar_rextra2}
\end{subfigure}
\begin{subfigure}{0.4\textwidth}
\includegraphics[width=\textwidth]{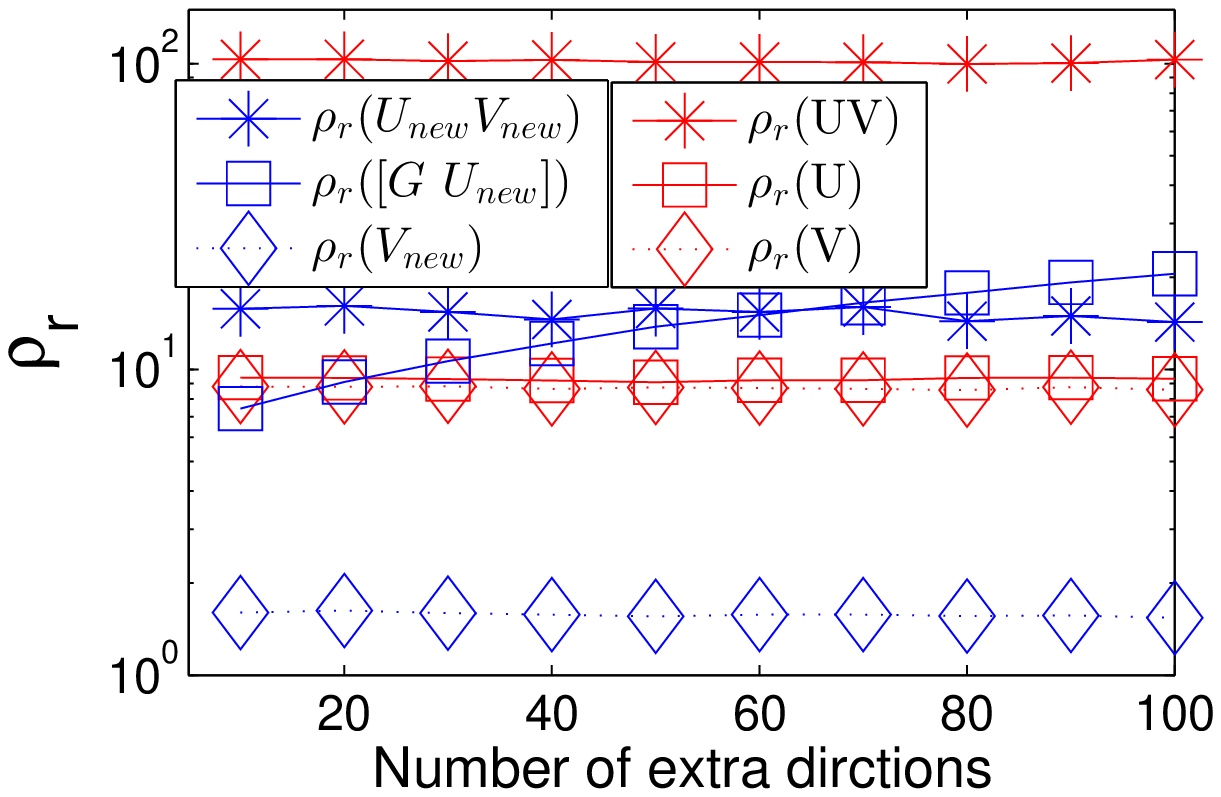}
\caption{Comparing the value of $\rho_r$}
\label{rhor_rextra2}
\end{subfigure}
\vspace{-0.1in}
\caption{\small{Comparison with increasing $r_\extra$ { ($n_1=200$, $d=200$, $n_2=120$, $m=0.075n_1n_2$, $r=20$, $r_0=18$, $r_\new=2$). In (b), we plot the value of $\rho_r$ needed to satisfy
(\ref{eq:PU}), (\ref{eq:PV}), (\ref{eq:UV}) and (\ref{eq:PU2}), (\ref{eq:PV2}), (\ref{eq:UV2}). We denote the respective values of $\rho_r$ by $\rho_r([\G\ \U_\new])$, $\rho_r(\V_\new)$, $\rho_r(\U_\new \V_\new)$, $\rho_r(\U)$, $\rho_r(\V)$ and $\rho_r(\U\V)$. Notice that $\rho_r(\U\V)$ is the largest, i.e. (\ref{eq:UV2}) is the hardest to satisfy.
Notice also that $\rho_r(\text{mod-PCP})=\max\{\rho_r([\G\ \U_\new]), \rho_r(\V_\new), \rho_r(\U_\new \V_\new)\}$ is significantly smaller than $\rho_r(\text{PCP})=\max\{\rho_r(\U), \rho_r(\V), \rho_r(\U\V)\}$. 
}}
}
\vspace{-0.2in}
\label{rextra2}
\end{figure}

In Fig. \ref{rnew}, we show comparisons with increasing number of new directions $r_\new$ (or equivalently decreasing $r_0 = r-r_\new$). We used $n_1=200$, $d=200$, $n_2=120$, $m=0.075n_1n_2$, $r=30$, $r_\extra=5$ and $r_\new$ ranging from $1$ to $20$ (thus $r_0$ ranges from 29 to 10). As we can see, mod-PCP performs better than PCP.

\begin{figure}[h!]
\centering
\includegraphics[width=0.45\textwidth]{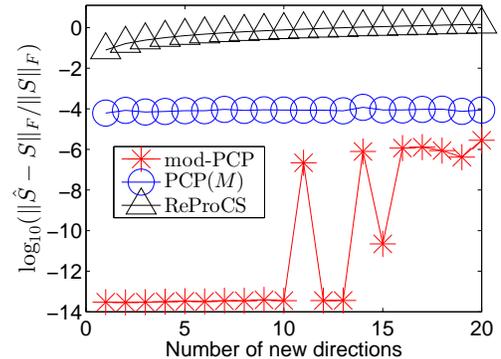} 
\caption{\small{Comparison with increasing $r_\new$  ($n_1=200$, $d=200$, $n_2=120$, $m=0.075n_1n_2$, $r=30$, $r_\extra=5$)}.}
\vspace{-0.2in}
\label{rnew}
\end{figure}


In Fig \ref{varc}, we show a comparison for increasing number of columns $n_2$. For this figure, we used $n_1=200, d=60,$ $r_G=r_0=18$, $r_\new = 2, m=0.075n_1n_2$, and $n_2$ ranging from 40 to 200. Notice that this is the situation where $n_2 \leq n_1$ so that $n_{(2)} = n_2$ and $n_{(1)} = n_1$. This situation typically occurs for time series applications, where one would like to use fewer columns to still get exact/accurate recovery. We compare mod-PCP and PCP. As we can see from Fig. \ref{spar_varc}, PCP needs many more columns than mod-PCP for exact recovery. Here we say exact recovery when $\|\Sb-\hat{\Sb}\|_F^2/\|\Sb\|_F^2$ is less than $10^{-6}$. Fig. \ref{rhor_varc} is the corresponding comparison of $\rho_r(\text{mod-PCP})$ and $\rho_r(\text{PCP})$ for this dataset and the conclusion is similar.

\begin{figure}[h!]
\centering
\begin{subfigure}{0.4\textwidth}
\includegraphics[width=\textwidth]{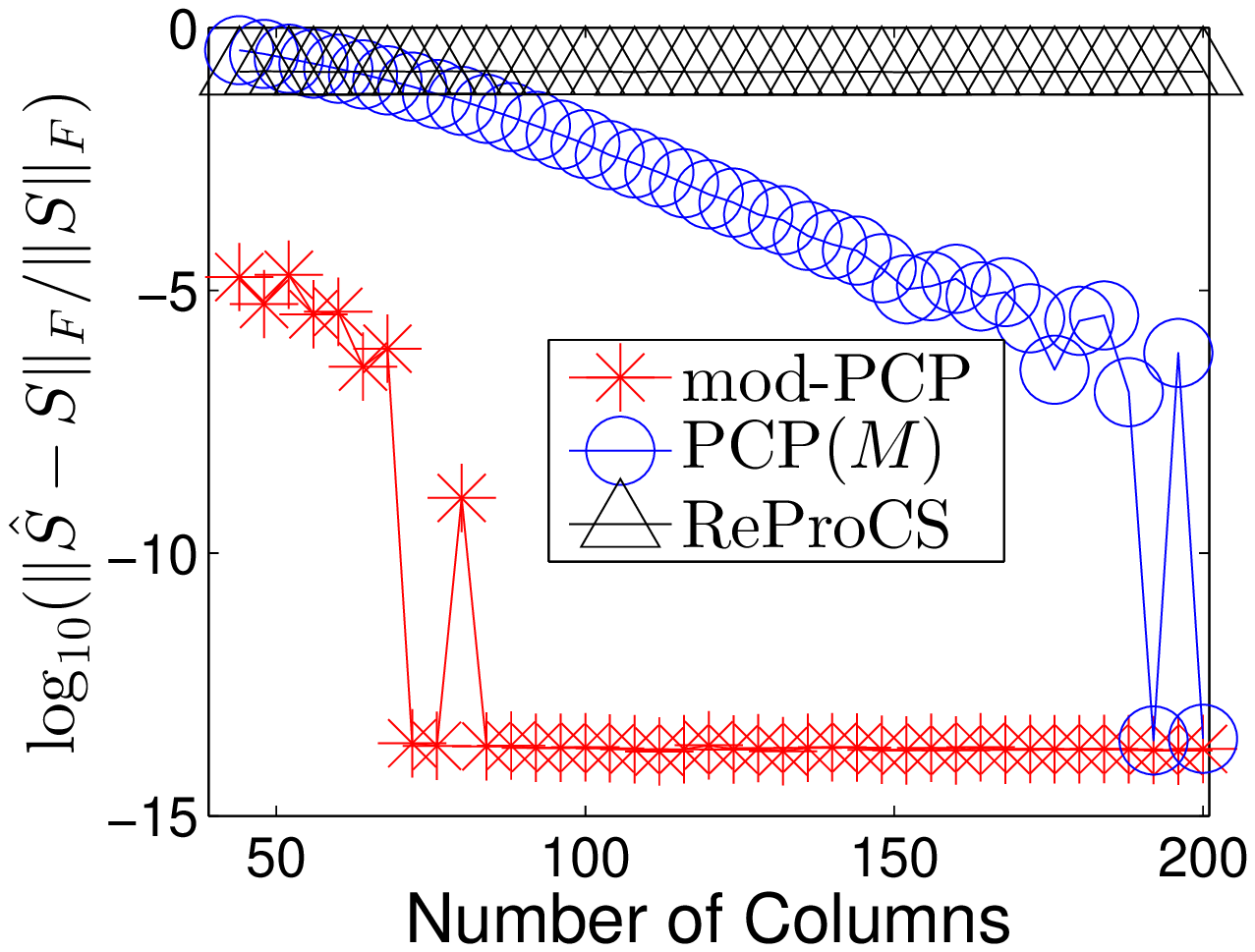}
\caption{Recovery result comparison}
\label{spar_varc}
\end{subfigure}
~
\begin{subfigure}{0.35\textwidth}
\centering
\includegraphics[width=\textwidth]{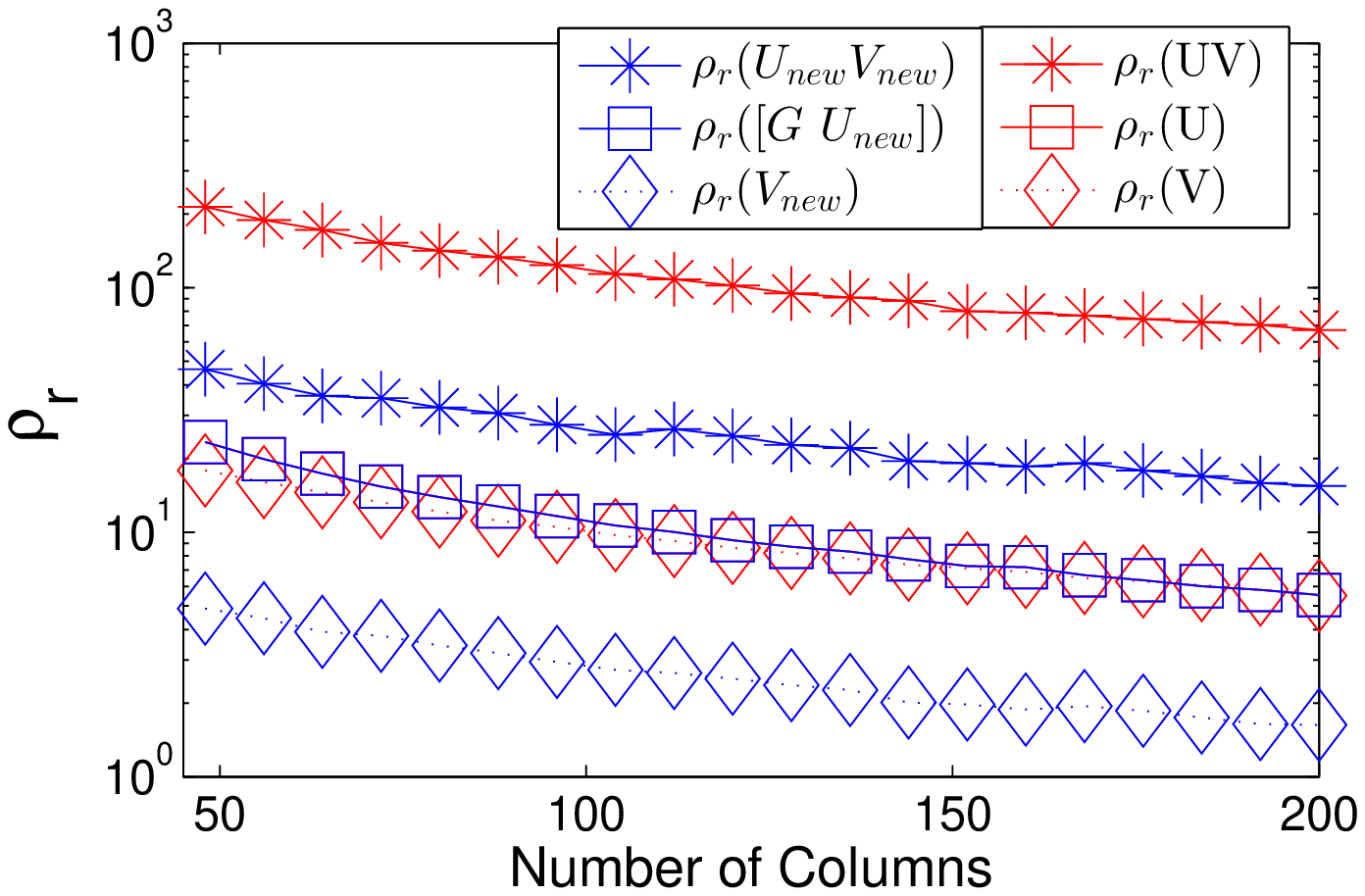}
\caption{Comparing the value of $\rho_r$}
\label{rhor_varc}
\end{subfigure}
\vspace{-0.1in}
\caption{\small{Comparison with increasing $n_2$ {($n_1=200, d=60,$ $r_G=r_0=18$, $r_\new = 2, m=0.075n_1n_2$).}}
}
\vspace{-0.2in}
\label{varc}
\end{figure}

Finally we generated phase transition plots similar to those for PCP in \cite{rpca}. We used the approach outlined in \cite{rpca} to generate $\Lb, \Sb$ and $\M$ i.e. we let $n_1=n_2=400$ and $\Lb=\X\Y^*$, where $\X$ and $\Y$ are independent $n_1\times r$ i.i.d. $\mathcal{N}(0,1/n_1)$ matrix and independent $n_2 \times r$ i.i.d. $\mathcal(0,1/n_2)$ matrices respectively. The support $\Omega$ of $\Sb$ is of size $m$ and uniformly distributed and for $(i,j)\in \Omega$, $\Pb(\Sb_{ij}=1)=\Pb(\Sb_{ij}=-1)=1/2$. For mod-PCP, we used $r_\new =\lfloor 0.15 r \rfloor$, $r_{\extra}=\lfloor0.15r\rfloor$ and we generated $\G$ as follows. We let $\U_0$ be the first $(r-r_\new)$ columns of the orthonormalized $\X$, and we generated $\G_\extra$ as the first $r_\extra$ columns of the orthonomalized $(\Ib-\U\U^*)\X_1$. Here $\U$ is the matrix of left singular vectors of $\Lb$ and $\X_1$ is a $n_1\times 2r_\extra$ i.i.d. $\mathcal{N}(0,1/n_1)$ matrix. We set $\G = [\U_0, \ \G_\extra]$.

To show the advantages of mod-PCP with less columns, we also did a comparison with the same parameters above but with $n_1=400, n_2=200$.
Fig. \ref{transition_rnew_015_rextra_015} shows the fraction of correct recoveries across 10 trials (as was also done in \cite{rpca}). Recoveries are considered correct if $\|\hat{L}-L\|_F/\|L\|_F\leq 10^{-3}$. 
As we can see from Fig. \ref{transition_rnew_015_rextra_015}, mod-PCP is always better than PCP since $r_\new$ and $r_\extra$ are small. But the difference is much more significant when $n_2 = n_1/2$ than when $n_2 = n_1$.
\begin{figure}[h!]
\centering
\begin{subfigure}[b]{0.2\textwidth}
\includegraphics[width=\textwidth]{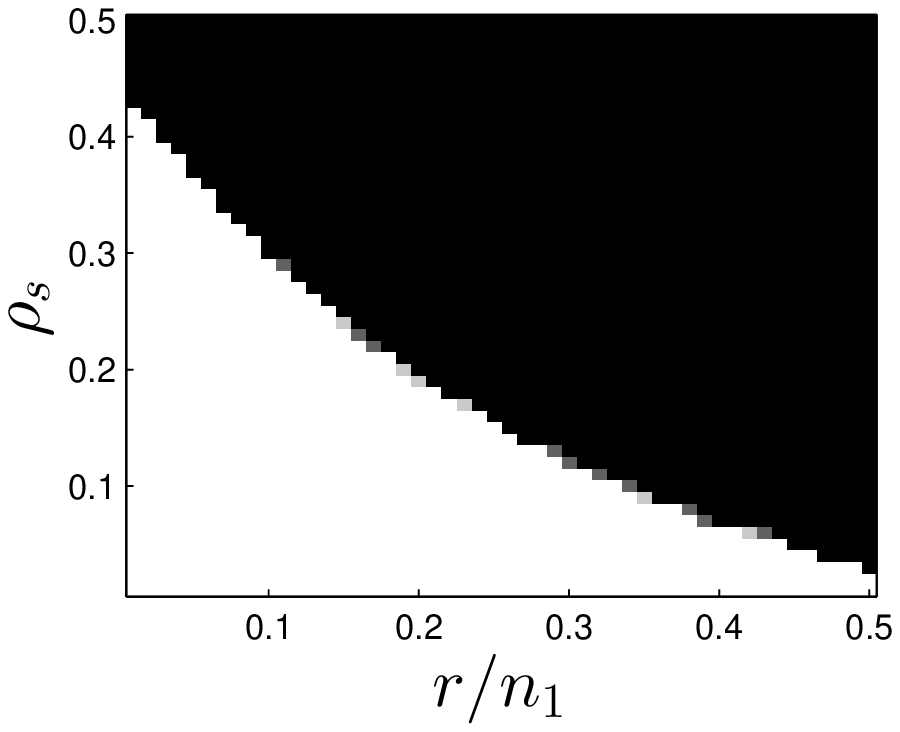}
\caption{mod-PCP,$n_2=400$}
\label{transition_n2_400_modpcp}
\end{subfigure}
~
\begin{subfigure}[b]{0.2\textwidth}
\includegraphics[width=\textwidth]{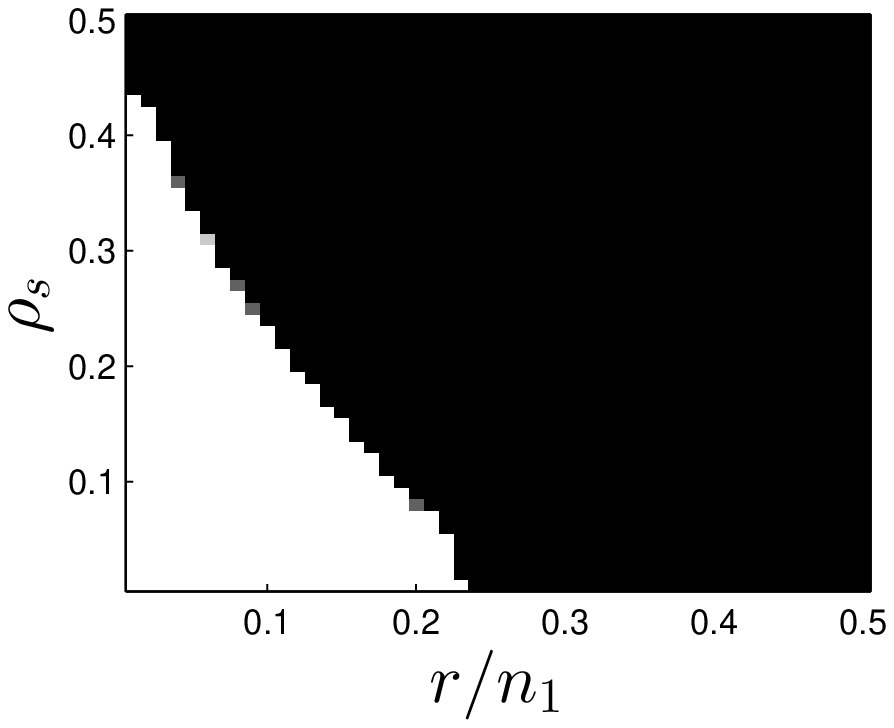}
\caption{PCP, $n_2=400$}
\label{transition_n2_400_pcp}
\end{subfigure}
~
\begin{subfigure}[b]{0.2\textwidth}
\includegraphics[width=\textwidth]{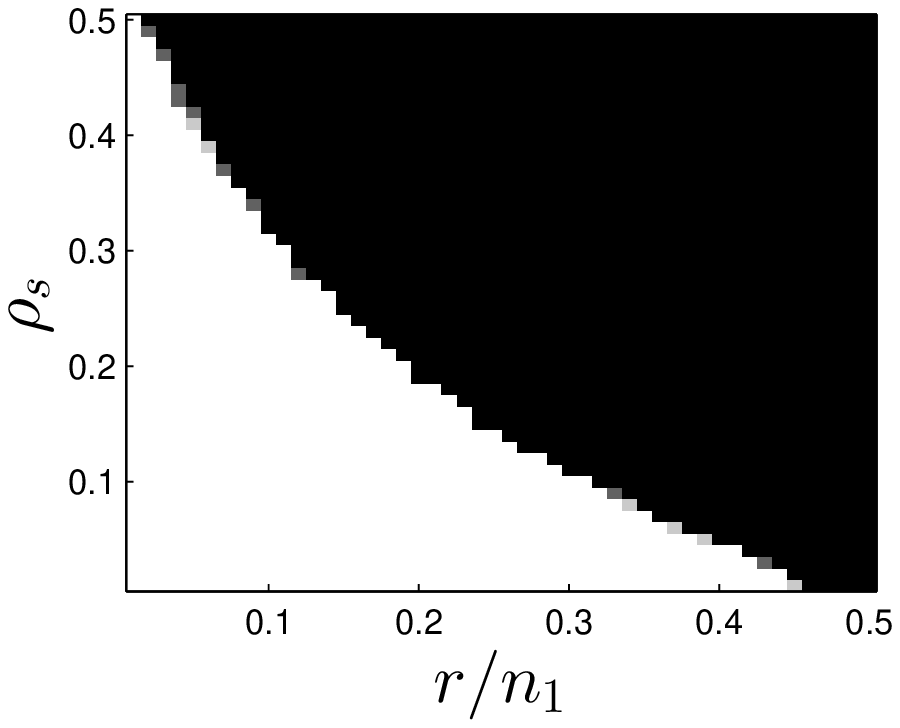}
\caption{mod-PCP, $n_2=200$}
\label{transition_n2_200_modpcp}
\end{subfigure}
~
\begin{subfigure}[b]{0.2\textwidth}
\includegraphics[width=\textwidth]{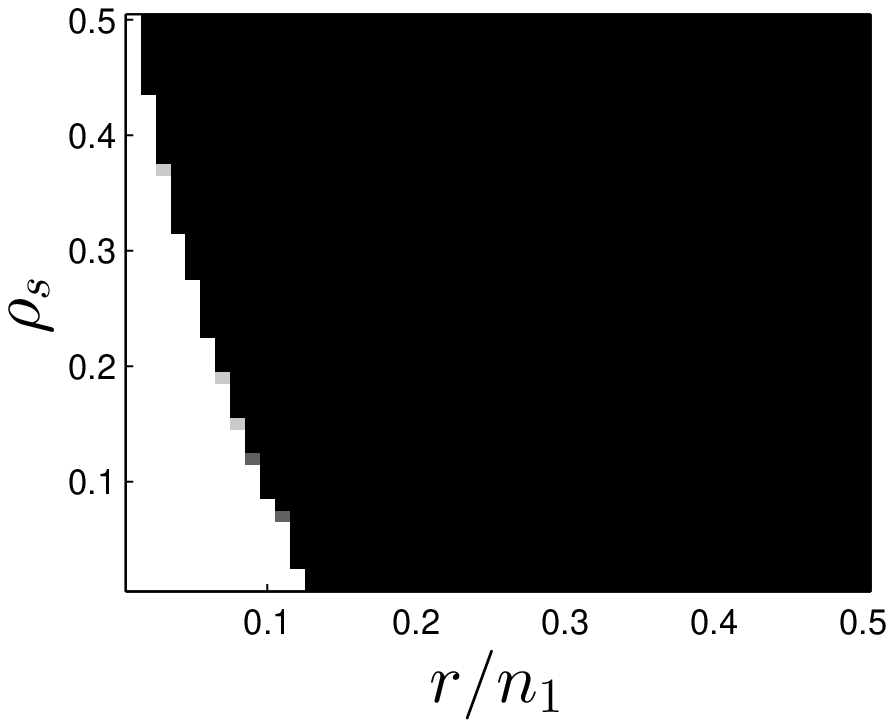}
\caption{PCP, $n_2=200$}
\label{transition_n2_200_pcp}
\end{subfigure}
\caption{Phase transition plots with $r_\new =\lfloor 0.15 r \rfloor$, $r_{\extra}=\lfloor0.15r\rfloor$, $n_1=400$}
\label{transition_rnew_015_rextra_015}
\vspace{-0.2in}
\end{figure}

\subsection{Real data (face reconstruction application)}

As stated in \cite{rpca}, robust PCA is useful in face recognition to remove sparse outliers, like cast shadows, specularities or eyeglasses, from a sequence of images of the same face. As explained there, without outliers, face images arranged as columns of a matrix are known to form an approximately low-rank matrix. Here we use the images from the Yale Face Database \cite{yaleface} that is also used in \cite{rpca}. Outlier-free training data consisting of face images taken under a few illumination conditions, but all without eyeglasses, is used to obtain a partial subspace estimate. The test data consists of face images under different lighting conditions and with eyeglasses or other outliers.
For test data, the goal is to reconstruct a clear face image with the cast shadows, eyeglasses or other outliers removed. Thus, the clear face image should be a column of the estimated low-rank matrix while the cast shadows or eyeglasses should be a column of the sparse matrix.

Each image is of size $243\times 320$, which we reduce to $122\times 160$. All images are re-arranged as long vectors and a mean image is subtracted from each of them. The mean image is computed as the empirical mean of all images in the training data.  For the training data, $\M_G$, we use images of subjects with no glasses, which is 12 subjects out of 15 subjects. We keep four face images per subject -- taken with center-light, right-light, left-light, and normal-light --  for each of these 12 subjects. Thus the training data matrix $\M_G$ is ${19520\times48}$. We compute $\G$ by keeping its left singular vectors corresponding to $99\%$ energy. This results in $r_G = 38$.
We use another two face images per subject for each of the twelve subjects, some with glasses and some without, as the test data, i.e. the measurement matrix $\M$. Thus $\M$ is ${19520\times24}$. 


In the experiments, we compare modified-PCP with PCP \cite{rpca} and ReProCS \cite{rrpcp_allerton,rrpcp_tsp} and also with some of the other algorithms compared in \cite{rrpcp_tsp}: robust subspace learning (RSL) \cite{Torre03aframework}, which is a batch robust PCA algorithm that was compared against in \cite{rpca}, and GRASTA \cite{grass_undersampled}, which is a very recent online robust PCA algorithm. We also compare against Dense Error Correction (DEC) \cite{error_correction_PCP_l1,wright2009robust} since this first addressed this application using $\ell_1$ minimization.
To implement Dense Error Correction (DEC) \cite{error_correction_PCP_l1,wright2009robust}, we normalize each column of $\M_G$ to get the dictionary $(\mathbf{D})_{n_1\times 48}$, and we solve
$$(\hat{\x}_i,\hat{\s}_i)=\arg \min_{\tilde{\x},\tilde{\s}} \|\tilde{\x}\|_1+\|\tilde{\s}\|_1 \text{ subject to } \M_i=\mathbf{D}\tilde{\x} + \tilde{\s}$$
using YALL-1. Here $\M_i$ is the $i$th column of $\M$. The solution gives us $\hat{\s}_i$ and $\hat{\lb}_i=\mathbf{D}\hat{\x}_i$.

For PCP and RSL, we use the test dataset only, i.e., $\M$, which is a ${19520\times24}$ matrix, as the measurement matrix. DEC, ReProCS and GRASTA are provided the same partial knowledge that mod-PCP gets.
Fig. \ref{realc} shows 3 cases where mod-PCP successfully removes the glasses into $(\hat{S})_i$ and gives the clearest estimate of the person's face without glasses as $(\hat{L})_i$. In the total 24 test frames, both mod-PCP and DEC remove the glasses (for those having glasses) or remove nothing (for those not having glasses) correctly in 14 of them, but the result of DEC has extra shadows in the face estimate. The other algorithms succeed for none of the 24 frames. Both ReProCS and GRASTA assume that the initial subspace estimate is accurate and ``slow subspace change" holds, neither of which happen here and this is the reason that neither of them work. RSL does not converge for this data set because the available number of frames is too small. The time taken by each algorithm is shown in Table \ref{speed}.

\begin{figure*}[h!]
\centering
\includegraphics[width=1.05\textwidth]{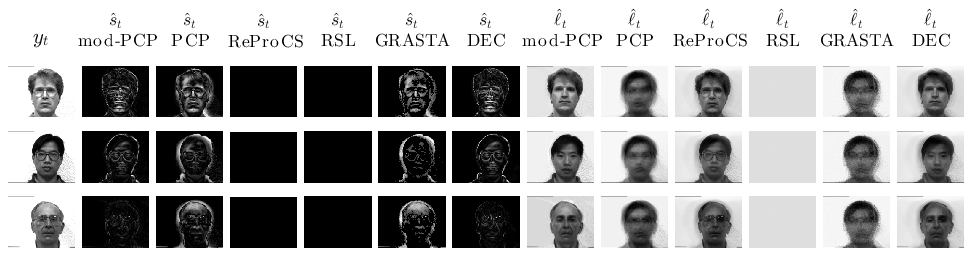}
\caption{Yale Face Image result comparison}
\label{realc}
\end{figure*}

\begin{table*}[h!]
		\centering
		\begin{tabular}{c|c|c|| c| c| c| c| c| c|c}
		\hline\hline
		DataSet & Image Size  & Sequence Length & mod-PCP&PCP&ReProCS&GRASTA&RSL&DEC&GOSUS \cite{xu2013gosus}\\
        \hline
        Yale Face & $122\times 160$ & 48 + 24 & 2.7 sec&9.8 sec&0.5 sec&50.2 sec&141.7 sec&21.3 sec\\
		\hline
		Lake & $72\times 90$ & $1420+80$ &  2.2 sec&1.7 sec&9.3 sec&338.7 sec&26.7 sec& \\
		\hline
        Fig. \ref{online_modpcp}&$256\times 1$&200+2400&2.7 sec&6.2 sec&12.0 sec&5.7 sec&25.4 sec&&576.9 sec\\
		\hline
		Fig. \ref{online_modpcp_reprocs}&$256\times 1$&200+8000&9.7 sec&18.9 sec&24.8 sec&12.6 sec&67.7 sec&&1735.6 sec\\
        \hline
        Fig. \ref{online_reprocs}&$256\times 1$&200+8000&13.1 sec&18.7 sec&26.1 sec&12.7 sec&74.8 sec&&1972.5 sec\\
		\hline
		\end{tabular}
		\caption{\small Speed comparison of different algorithms. { (Sequence length refers to the length of sequence for training plus the length of sequence.)} 
        \label{speed}
}
	\end{table*}


\subsection{Online robust PCA: simulated data comparisons}
For simulation comparisons for online robust PCA, we generated data as explained in \cite{reprocs_correctness_result}. The data was generated using the model given in Section \ref{online_robust_pca}, with $n=256$, $J=3$, $r_0=40$, $t_0=200$ and $c_{j,\new}=4$, $c_{j,\old}=4$, for each $j=1,2,3$. The coefficients, $\ab_{t,*}=\Pbf_{j-1}^*\lb_t$ were i.i.d. uniformly distributed in the interval $[-\gamma,\gamma]$; the coefficients along the new directions, $\ab_{t,\new} := \Pbf_{j,\new}^*\lb_t$ generated i.i.d. uniformly distributed in the interval $[-\gamma_\new, \gamma_\new]$  (with a $\gamma_\new \le \gamma$) for the first $1700$ columns after the subspace change and i.i.d. uniformly distributed in the interval $[-\gamma,\gamma]$ after that. We vary the value of $\gamma_\new$; small values mean that ``slow subspace change" required by ReProCS holds. The sparse matrix $\Sb$ was generated in two different ways to simulate uncorrelated and correlated support change.
For partial knowledge, $\G$, we first did SVD decomposition on $[\lb_1, \lb_2, \cdots, \lb_{t_0}]$ and kept the directions corresponding to singular values larger than $\mathbf{E}(z^2)/9$, where $z\thicksim \text{Unif}[-\gamma_\new, \gamma_\new]$. We solved PCP and modified-PCP every $200$ frames by using the observations for the last 200 frames as the matrix $\M$.
The ReProCS algorithm of \cite{rrpcp_perf,reprocs_correctness_result} was implemented with $\alpha=100$. The averaged sparse part errors with three different sets of parameters over 20 Monte Carlo simulations are displayed in Fig. \ref{online_modpcp}, Fig. \ref{online_modpcp_reprocs}, and Fig. \ref{online_reprocs}, and the corresponding averaged time spent for each algorithm is shown in Table \ref{speed}.  For all three figures, we used $t_1 = t_0+6\alpha+1$, $t_2 = t_0+12\alpha+1$ and $t_3 = t_0+18\alpha+1$ and $\gamma = 5$.

In the first case, Fig. \ref{online_modpcp}, we used $\gamma_\new = \gamma$ and so ``slow subspace change" does not hold. For the sparse vectors $\s_t$, each index is chosen to be in support with probability $0.0781$. The nonzero entries are uniformly distributed between $[20, 60]$. Since ``slow subspace change" does not hold, ReProCS does not work well. Since the support is generated independently over time, this is a good case for both PCP and mod-PCP. Mod-PCP has the smallest sparse recovery error.
In the second case, Fig. \ref{online_modpcp_reprocs}, we used $\gamma_\new = 1$ and thus ``slow subspace change" holds. For sparse vectors, $\s_t$, the support is generated in a correlated fashion. We used support size $s=5$ for each $\s_t$; the support remained constant for 25 columns and then moved down by $s=5$ indices. Once it reached $n$, it rolled back over to index one. Because of the correlated support change, PCP does not work. In this case, both mod-PCP and ReProCS work but PCP does not. 
In the third case, Fig. \ref{online_reprocs}, the parameters are the same as in the second case, except that the support size is $s=10$ in each column and it moves down by $s/2 = 5$ indices every 25 columns. In this case, the sparse vectors are much more correlated over time, resulting in sparse matrix $\Sb$ that is even more low rank, thus neither mod-PCP nor PCP work for this data. In this case, only ReProCS works.

Thus from simulations, modified-PCP is able to handle correlated support change better than PCP but worse than ReProCS. Modified-PCP also works when slow subspace change does not hold; this is a situation where ReProCS fails. Of course, modified-PCP, GRASTA and ReProCS are provided the same partial subspace knowledge $\G$ while PCP and RSL do not get this information.

\begin{figure*}[h!]
\centerline{
\begin{subfigure}[b]{0.33\textwidth}
\includegraphics[width=0.9\linewidth]{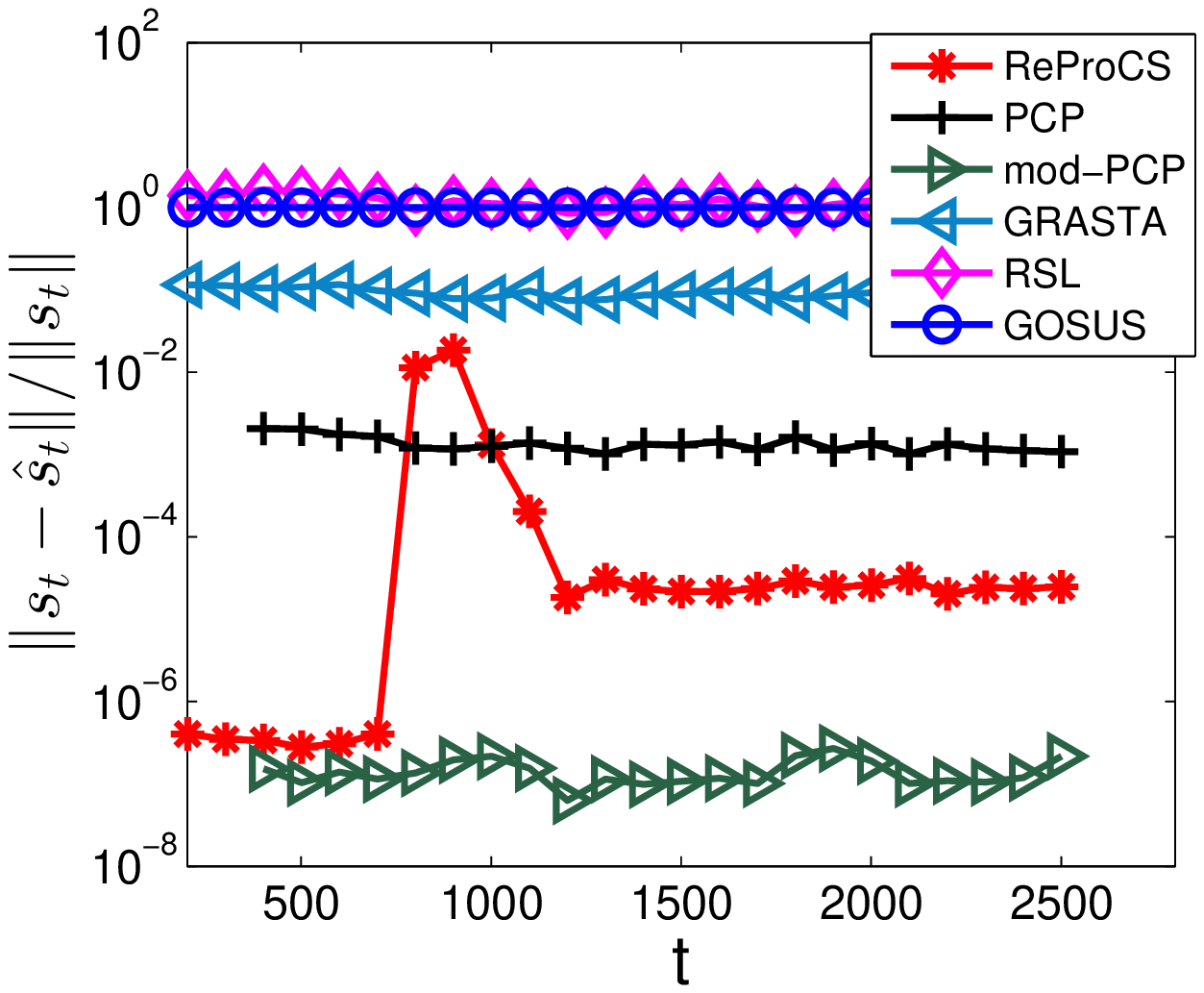}
\caption{Uniform distributed $\s_t$}
\label{online_modpcp}
\end{subfigure}
~
\begin{subfigure}[b]{0.33\textwidth}
\includegraphics[width=0.9\linewidth]{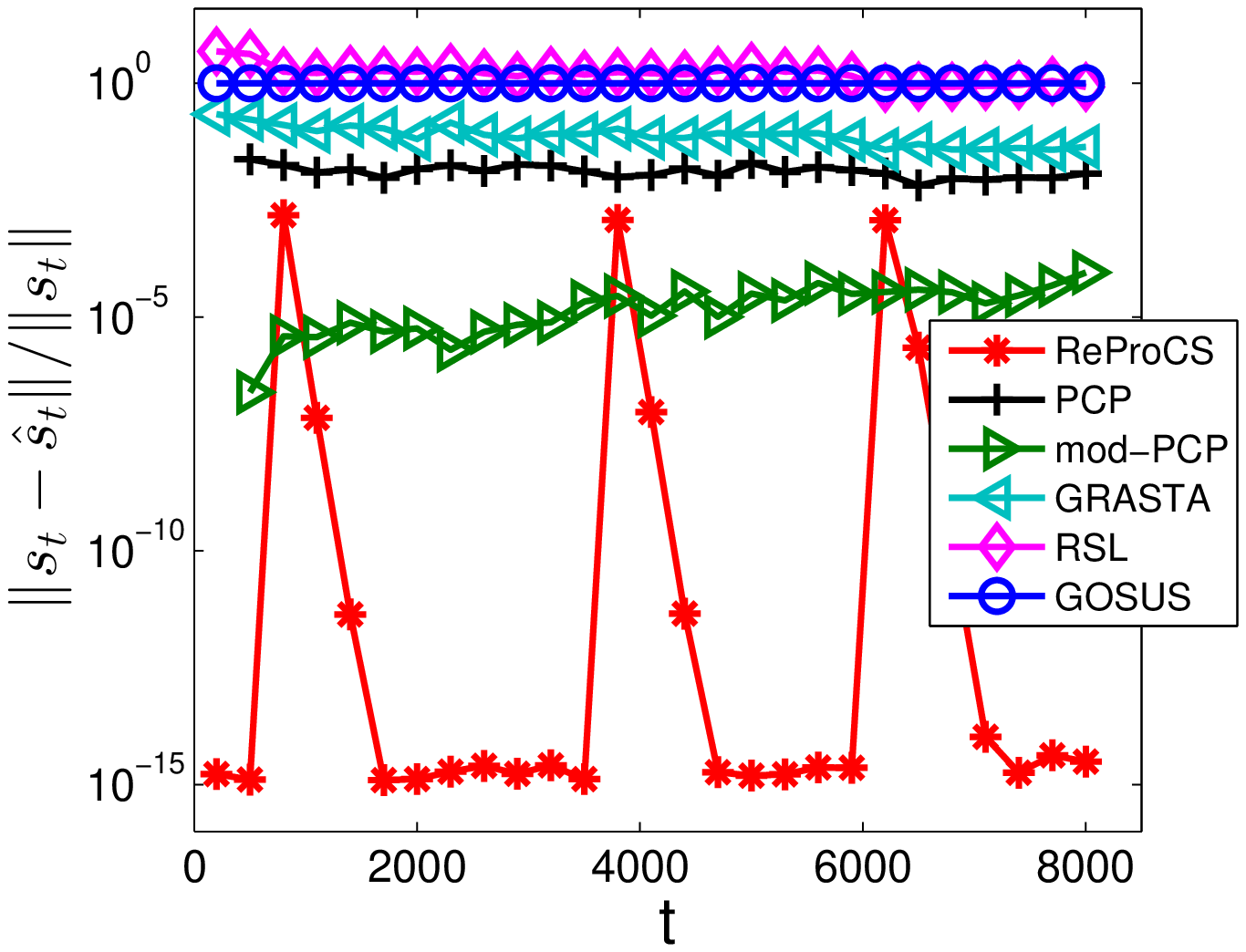}
\caption{Correlated $\s_t$ with small support size}
\label{online_modpcp_reprocs}
\end{subfigure}
~
\begin{subfigure}[b]{0.33\textwidth}
\includegraphics[width=0.9\linewidth]{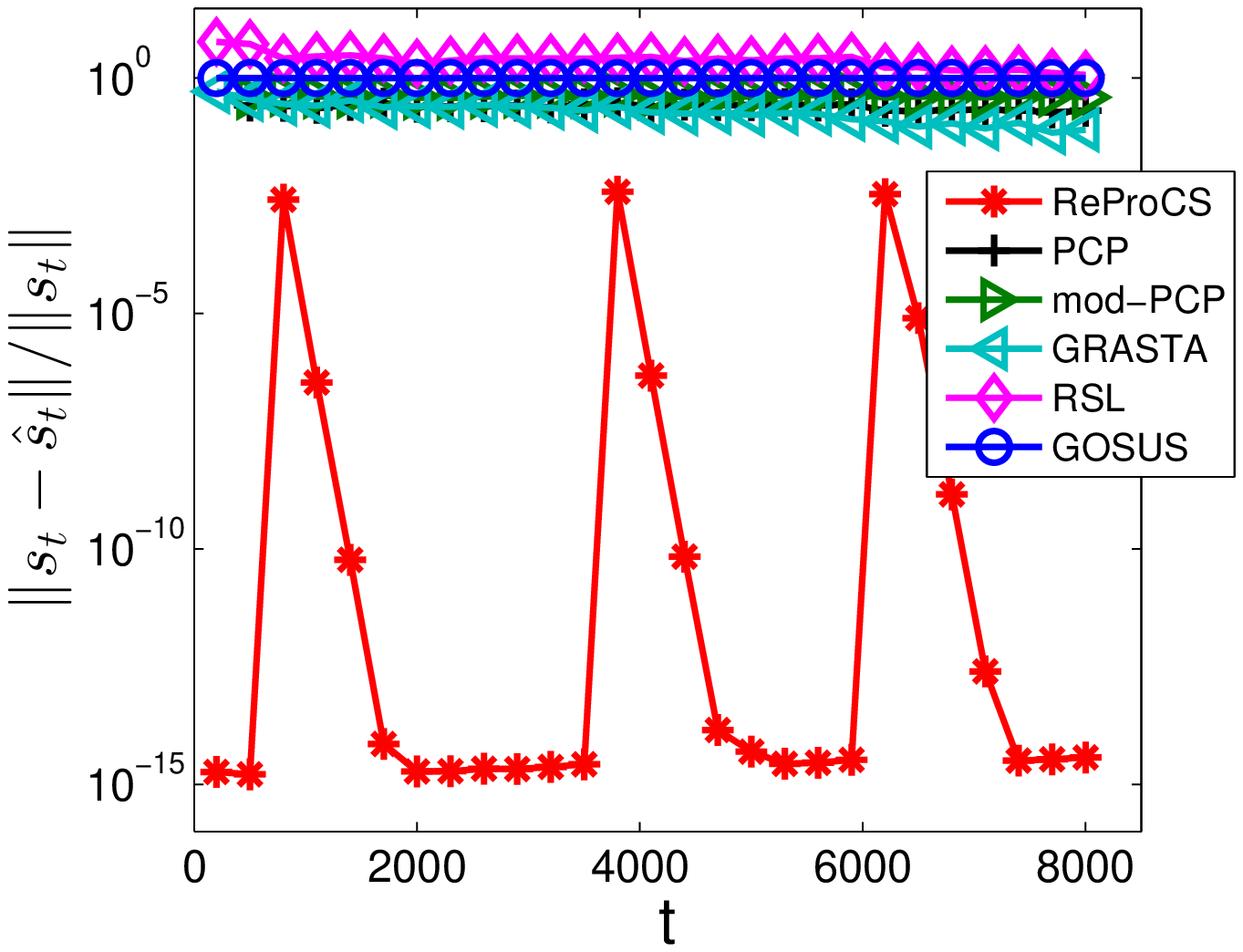}
\caption{Correlated $\s_t$ with large support size}
\label{online_reprocs}
\end{subfigure}
}
\caption{\small NRMSE of sparse part comparison with online model { ($n=256$, $J=3$, $r_0=40$, $t_0=200$, $c_{j,\new}=4$, $c_{j,\old}=4$, $j=1,2,3$)}
}
\label{online}
\vspace{-0.1in}
\end{figure*}

\subsection{Online robust PCA: comparisons for video layering} \label{real_online}
The lake sequence is similar to the one used in \cite{rrpcp_tsp}. The background consists of a video of moving lake waters. The foreground is a simulated moving rectangular object.
The sequence is of size $72\times 90 \times 1500$, and we used the first $1420$ frames as training data (after subtracting the empirical mean of the training images), i.e. $\M_G$. The rest 80 frames (after subtracting the same mean image) served as the background $\Lb$ for the test data. For the first frame of test data, we generated a rectangular foreground support with upper left vertex $(1,j_0)$ and lower right vertex $(i_1, 25+j_0)$, where $j_0 \sim \text{Unif}[1,30]$ and $i_1 \sim \text{Unif}[7,16]$, and the foreground moves to the right 1 column each time. Then we stacked each image as a long vector $\lb_t$ of size $6480\times1$. For each index $i$ belonging to the support set of foreground $\s_t$, we assign $(\s_t)_i = 185 - (\lb_t)_i$. We set $\M = \Lb + \Sb$.
For mod-PCP, ReProCS and GRASTA, we used the approach used in \cite{rrpcp_tsp} to estimate the initial background subspace (partial knowledge): do SVD on $\M_G$ and keep the left singular vectors corresponding to $95\%$ energy as the matrix $\G$. A few recovered frames are shown in Fig. \ref{lake}, and the averaged normalized mean squared error (NMSE) of the sparse part over $50$ Monte Carlo realizations is shown in Fig. \ref{lake_nmse}. The averaged time spent for each algorithm is shown in Table \ref{speed}. As can be seen, in this case, both mod-PCP and ReProCS perform almost equally well, with ReProCS being slightly better.

Next we compute the value of $\rho_r$ for the lake video sequence. 
We calculated prior knowledge $\G$ as explained above.  We calculated the singular vectors $\U, \V$ by doing SVD decomposition on ${\Lb}$ and keeping all the directions with corresponding singular values larger than $10^{-10}$ (we choose $10^{-10}$ because it is the precision that MATLAB can achieve for SVD decomposition); calculate $\U_{\new}, \V_{\new}$ by doing SVD decomposition of $(I-\G\G^*){\Lb}$ and keeping all the directions with singular values larger than $10^{-10}$. With this, we get $\rho_r(\text{PCP})= 1.8584\times10^4$ and $\rho_r(\text{mod-PCP})= 1.7785\times10^4$.

We also calculate $\rho_r$ for fountain02 sequence, which can be found on http://changedetection.net/. The image size is $288\times432$, and we resize it to $96\times144$. For the first 600 background images we form a low rank matrix $[\M_G\ \Lb]$ by stacking each image as a column (the first 300 columns belong to $\M_G$ and the rest belong to $\Lb$). With the same steps for lake sequence, we get $\rho_r$(PCP) is $4.311\times10^4$ and $\rho_r$(mod-PCP) is $1.7866\times10^4$.

\begin{figure*}[h!]
\centering
\includegraphics[width=1.03\linewidth]{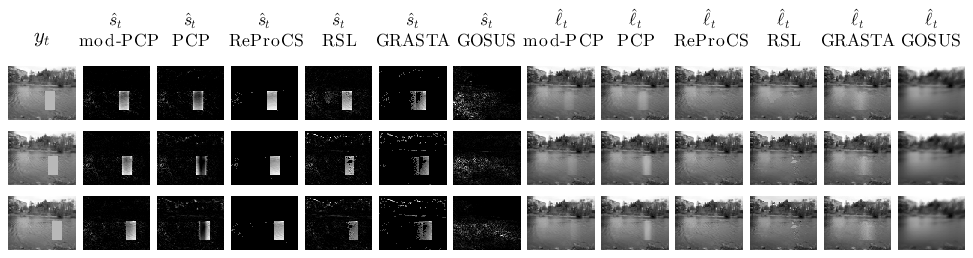}
\caption{\small{Lake sequence result comparison (columns $60,69,79$ are shown here. Note that in the last 2 rows, clearly there is missing part in $\s_t$ and corresponding extra part in $\lb_t$ the back detected by RSL)}.
}
\label{lake}
\vspace{-0.1in}
\end{figure*}

\begin{figure}[h!]
\centering
\includegraphics[width=0.9\linewidth]{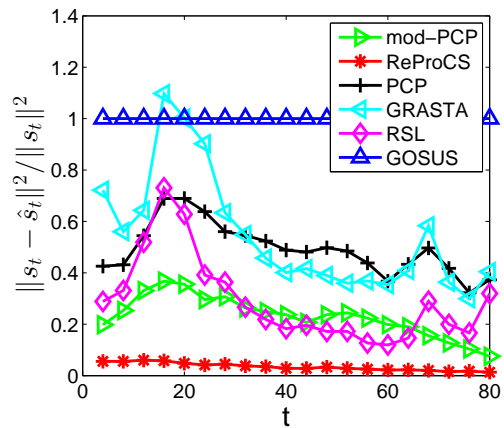}
\caption{\small{Lake sequence NMSE comparison.} 
}
\label{lake_nmse}
\vspace{-0.2in}
\end{figure}

\subsection{Comparison with Simulated Noisy Data}
In order to address an anonymous reviewer's comment, we have also added simulations with noisy data. We assume the measurement model
\begin{equation}
\M=\Lb+\Sb+\Z
\end{equation}
where $\Lb$ is low rank (with partial knowledge $\G$ similar to previous case), $\Sb$ is sparse and $\Z$ is a noise term with $\|\Z\|_F \leq \sigma$. Inspired by \cite{stablepcp}, we propose the following optimization problem to solve the problem:
\begin{equation}
\label{eq:sdp_noisy}
  \begin{array}{ll}
    \text{minimize}_{\tL_\new,\tS, \tX}   & \quad \|\tL_\new\|_* + \lambda \|\tS\|_1\\
    \text{subject to} & \quad \|\tL_\new + \G\tX^* + \tS -  \M \|_F \leq \sigma
  \end{array}
\end{equation}
with $\lambda = \sqrt{\max\{n_1, n_2\}}$.
To compare the result with stable PCP \cite{stablepcp}, we generated square matrices as stated in \cite[Section V]{stablepcp}, i.e., $n_1=n_2=200$, $r=10$, $r_\new=2$, $r_\extra=0$, $\rho_s=0.2$, $\Lb=\X\Y^*$ where $\X$ and $\Y$ are independent $n_1\times r$ i.i.d. $\mathcal{N}(0,1/n_1)$ matrices, and each entry of $\Sb$ is independently distributed, taking value $0$ with probability $1-\rho_s$ and uniformly distributed in $[-5,5]$ with probability $\rho_s$. We used the same suggested $\bar{\tau}$ for the stable mod-PCP solver as in \cite{stablepcp}. By varying $\sigma$ from $0.1$ to $1$, we got recovery errors over $50$ Monte Carlo simulations as shown in Fig. \ref{noisy_var_sigma}.
We plot the root-mean-squared (RMS) error which is defined in \cite{stablepcp} as the average of $\|\hat{\Lb}-\Lb\|_F/n$ for the low-rank matrix and of $\|\hat{\Sb}-\Sb\|_F/n$ for the sparse matrix.

\begin{figure}[h!]
\centering
\includegraphics[width=0.8\linewidth]{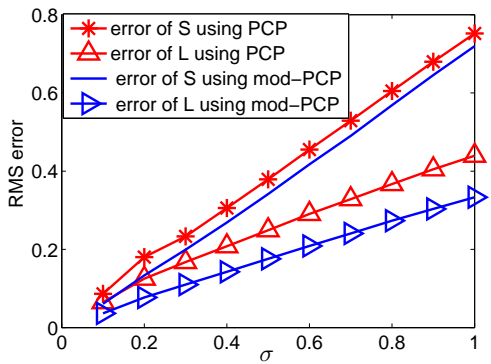}
\caption{\small{Noisy data RMS error comparison.}}
\label{noisy_var_sigma}
\end{figure}

\section{Conclusions}
In this work we studied the following problem. Suppose that we have a partial estimate of the column space of the low rank matrix $\Lb$. How can we use this information to improve the PCP solution? We proposed a simple modification of PCP, called {\em modified-PCP}, that allows us to use this knowledge. We derived its correctness result that allows us to argue that, when the available subspace knowledge is accurate enough, modified-PCP requires significantly weaker incoherence assumptions on the low-rank matrix than PCP. We also obtained a useful corollary (Corollary \ref{cor}) for the online or recursive robust PCA problem. Extensive simulation experiments and some experiments for a real application further illustrate these claims. Ongoing work includes studying the error stability of modified-PCP for online robust PCA. 
Future work will include developing a fast and recursive algorithm for solving modified-PCP and using the resulting algorithm for various practical applications. Two applications that will be explored are (a) video layering, e.g. using the BMC dataset of \cite{bouwmans2014robust}, and (b) recommendation system design in the presence of outliers and missing data.
For getting a recursive algorithm, we will explore the use of ideas similar to those introduced in Feng et al's recent work on developing a recursive algorithm that asymptotically approximates the PCP solution \cite{xu_nips2013_2}.

\appendix

\subsection{Derivation for (\ref{eqG})}
\label{probdef_proof}

Recall from Sec \ref{probdef} that $r_\new = \rank(\Lb_\new)$,
\begin{equation}
\Lb_\new=(\Ib-\G\G^*)\Lb  \stackrel{\text{SVD}}{=} \U_\new \Sigmab_\new \V_\new^*
\end{equation}
Let $\U_0$ be a basis matrix for $\Span(\Lb) \cap \Span(\G) = \Span(\U) \cap \Span(\G)$ with $r_0= \rank(\U_0)$
Thus, there exist rotation matrices $\Rb_1, \Rb_G$ and basis matrices $\U_1, \G_\extra$ such that
\begin{equation} \label{defUG1}
\U \Rb_1 = [\U_0 \ \U_1] \ \text{and} \ \G \Rb_G = [\U_0 \ \G_\extra]
\end{equation}
with $\G_\extra{}^* \U_1 = 0$. 

Clearly, $\rank(\U_1) = r_\new$ \footnote{This follows because $(\Ib- \G \G^*) \Lb = (\Ib - \U_0 \U_0^*) [\U_0 \ \U_1] \Rb_1^{-1} \Sigmab \V^* = [0 \ \U_1] \Rb_1^* \Sigmab \V^*$. Since $\rank([0 \ \U_1]) = \rank(\U_1)$ and all other matrices are full rank $r$, we get that $\rank(\U_1) = \rank(\Lb_\new)=r_\new$. Here we have used Sylvester's inequality on $\Lb_\new = [0 \ \U_1] (\Rb_1^* \Sigmab \V^*)$ to get that $\rank(\U_1) +r - r \le \rank(L_\new)=r_\new \le \min(\rank(\U_1),r)=\rank(\U_1)$.}.
Split the $r\times r$ matrix $\Rb_1$ as $\Rb_1 = [(\Rb_1)_0, \ (\Rb_1)_1]$ so that $(\Rb_1)_0$ contains the first $r_0$ columns and $(\Rb_1)_1$ contains the last $r_\new$ columns. Thus, 
$$\Lb_\new = (\Ib - \U_0 \U_0^*) [\U_0 \ \U_1] \Rb_1^* \Sigmab \V^* =  \U_1 (\Rb_1)_1^* \Sigmab \V^*.$$
Let $((\Rb_1)_1^* \Sigmab \V^*) \stackrel{\text{SVD}}{=} \U_2 \Sigmab_2 \V_2^*$ denote its full SVD. Thus $\Lb_\new = \U_1 \U_2 \Sigmab_2 \V_2^*$. Comparing with the SVD of $\Lb_\new$ we get that
$\U_\new = \U_1 \U_2$ where $\U_2$ is a $r_\new \times r_\new$ unitary matrix; $\Sigmab_\new = \Sigmab_2$ and $\V_\new = \V_2$.
Thus,
\begin{eqnarray}
\U \Rb_1 = [\U_0 \ \U_\new \U_2^*] = [\U_0 \ \U_\new ]\left(\Ib\ \ \  \  \zero \atop \zero \ \ \U_2^*\right)
\end{eqnarray}
By taking $\Rb_U =\Rb_1 \left(\Ib\ \ \  \  \zero \atop \zero \ \ \U_2^*\right)^{-1}  = \Rb_1\left(\Ib\ \ \ \ \zero\atop\zero\ \ \U_2\right)$, we get
\begin{equation} \label{defUG}
\U \Rb_U = [\U_0 \ \U_\new] \ \text{and} \ \G \Rb_G = [\U_0 \ \G_\extra]
\end{equation}
Rearranging, we get (\ref{eqG}).

\subsection{Proof of Lemma \ref{unifber}} \label{proofunifber}
First we state and prove the following fact\footnote{This fact may seem intuitively obvious, however we cannot find a simpler proof for it than the one we give.}.
\begin{proposition}
\label{unif}
Assume $m_1 < m_2<n_1n_2$, we have $$\Pb_{\text{Unif}(m_1)}(\text{Success})\geq \Pb_{\text{Unif}(m_2)}(\text{Success}).$$
\end{proposition}

There are a total of $n_1 n_2 \choose {m_2}$ size-$m_2$ subsets of the set of indices of an $n_1 \times n_2$ matrix. The probability of any one of them getting selected is $1/{n_1 n_2 \choose {m_2}}$ under the $\text{Unif}(m_2)$ model. Suppose that the algorithm succeeds for $k$ out of these $n_1 n_2 \choose {m_2}$ sets. Call these the ``good" sets. Then,
$$\Pb_{\text{Unif}(m_2)}(\text{Success}) = \frac{k}{ {n_1 n_2 \choose {m_2}} }.$$
By Theorem 2.2 of \cite{rpca}, the algorithm definitely also succeeds for all size-$m_1$ subsets of these $k$ ``good" size-$m_2$ sets. Let $k_1$ be the number of such size $m_1$ subsets.
Under the $\text{Unif}(m_1)$ model, the probability of any one such set getting selected is $\frac{1}{{n_1 n_2 \choose {m_1}}}$. Thus
$\Pb_{\text{Unif}(m_1)}(\text{Success}) = \frac{k_1}{ {n_1 n_2 \choose {m_1}} }.$

Now we need to lower bound $k_1$. 
There are a total of $n_1 n_2 \choose {m_2}$ size-$m_2$ sets and each of them has $m_2 \choose {m_1}$ subsets of size $m_1$. However, the total number of distinct size-$m_1$ sets is only $n_1 n_2 \choose {m_1}$. Because of symmetry, this means that in the collection of all size-$m_1$ subsets of all size-$m_2$ sets, a given set is repeated $b=\frac{ {n_1 n_2 \choose {m_2}}  {{m_2} \choose {m1}} }{ {n_1 n_2 \choose {m_1}} }$ times.

In the sub-collection of  size-$m_1$ subsets of the $k$ ``good" size-$m_2$ sets, the number of times a set is repeated is less than or equal to $b$. Also, the number of entries in this collection (including repeated ones) is $k {m_2 \choose m_1}$.
Thus, the number of distinct size-$m_1$ subsets of the ``good" sets is lower bounded by $\frac{k {m_2 \choose {m_1}}}{b}$, i.e. $k_1 \ge \frac{k {m_2 \choose m_1} }{b}.$ 
Thus,
$$\Pb_{\text{Unif}(m_1)}(\text{Success})  \ge \frac{ k {{m_2} \choose {m_1}} {{n_1 n_2} \choose {m_1}} }{ {{n_1 n_2} \choose {m_1}} {{m_2} \choose {m_1}} {{n_1 n_2} \choose {m_2}} } = \Pb_{\text{Unif}(m_2)}(\text{Success}).$$ 


\begin{proof}[Proof of Lemma \ref{unifber}]
Denote by $\Omega_0$ the support set. We have
\begin{align*}
  &\Pb_{\text{Ber}(\rho_0)}(\text{Success}) \\
  =& \sum_{k = 0}^{n_1n_2} \Pb_{\text{Ber}(\rho_0)}(\text{Success} \, | \, |\Omega_0| = k) \Pb_{\text{Ber}(\rho_0)}(|\Omega_0| = k)\\
  \le& \sum_{k = 0}^{m_0-1} \Pb_{\text{Ber}(\rho_0)}(|\Omega_0| = k) + \\
  &\sum_{k = m_0}^{n_1n_2}\Pb_{\text{Unif}(k)}(\text{Success}) \Pb_{\text{Ber}(\rho_0)}(|\Omega_0| = k)\\
  \le& \Pb_{\text{Ber}(\rho_0)}(|\Omega_0| < m_0) +
  \Pb_{\text{Unif}(m_0)}(\text{Success}),
\end{align*}
where we have used the fact that for $k \ge m_0$, $\Pb_{\text{Unif}(k)}(\text{Success} ) \le \Pb_{\text{Unif}(m_0)}( \text{Success} )$ by Proposition \ref{unif}, and that
the conditional distribution of $\Omega_0$ given its cardinality is
uniform. Thus,
\[
\Pb_{\text{Unif}(m_0)} (\text{Success}) \ge
\Pb_{\text{Ber}(\rho_0)}(\text{Success}) - \Pb_{\text{Ber}(\rho_0)}(|\Omega_0| <
m_0).
\]
Let random matrix $\X^{n_1\times n_2}$ be a matrix whose each entry is i.i.d. Bernoulli distributed as $\Pb(\X_{ij}=1)=\rho_0, \Pb(\X_{ij}=0)=1-\rho_0$. Then, under the Bernoulli model,  $|\Omega_0|=\sum_{i,j}\X_{ij}$,  $\E [\sum_{i,j}\X_{ij}] = \E[|\Omega_0|] =\rho_0n_1n_2$, and $0 \leq \X_{ij} \leq 1$. Thus by the Hoeffding inequality, we have
$$\Pb(\E[ \sum_{i,j}\X_{ij}] - \sum_{i,j}\X_{ij} \geq t) \leq \exp(-\frac{2t^2}{n_1n_2}).$$
As $\rho_0=\frac{m_0}{n_1n_2}+ \epsilon_0$, take $t=\epsilon_0 n_1n_2$, we have
$$ \Pb_{\text{Ber}(\rho_0)}(|\Omega_0| \leq m_0) =  \Pb(\sum_{i,j}\X_{ij} \leq m_0) \leq \exp(-2\epsilon_0^2n_1n_2).$$
Thus
$
\Pb_{\text{Unif}(m_0)} (\text{Success}) \ge \Pb_{\text{Ber}(\rho_0)}(\text{Success}) - \exp(-2\epsilon_0^2n_1n_2).
$
\end{proof}

\subsection{Proof of Lemma \ref{sign2norm}}
\label{proofsign2norm}

\begin{proof}
First, we state the theorem used in this proof.

\begin{lemma}\cite[Theorem 2(10a)]{furedi1981eigenvalues}
\label{FurediLemma}
For $n\times n$ matrix $\A$ with entries $a_{ij}$, let $a_{ij}, i\geq j$ be independent (not necessarily identically distributed) random variables bounded with a common bound $K$. Assume that for $i\geq j$, the $a_{ij}$ have a common expectation $\mu=0$ and variance $\sigma^2$. Define $a_{ij}$ for $i<j$ by $a_{ij}=a_{ji}.$ (The numbers $K, \mu, \sigma^2$ will be kept fixed as the matrix dimension $n$ will tend to infinity.) For $k$ satisfying $K^2k^6/(4\sigma^2n)<1/2$, we have
$$\Pb(\max_i(|\lambda_i(\A)|)>2\sigma\sqrt{n}+v)<\sqrt{n}\exp(-\frac{kv}{2\sigma\sqrt{n}+v}).$$
\end{lemma}
Proof: see Appendix \ref{proofFurediLemma}.
This is a minor modification of the upper bound of \cite[Theorem 4]{achlioptas2001fast}, \cite[Theorem 1.4]{vu2005spectral}. The only change is that it allows the variance of $a_{ij}$ to be bounded by $\sigma^2$ instead of forcing it to be equal to $\sigma^2$.

Let
\begin{equation}
\A:=\left(\begin{array}{cc} 0&\Eb\\\Eb^*&0\end{array}\right)
\end{equation}
Notice that $\A$ is an $(n_1+n_2) \times (n_1+n_2)$ symmetric matrix that satisfies requirements of Lemma \ref{FurediLemma}. 
By Lemma \ref{FurediLemma} with $K=1, \mu=0, \sigma=\sqrt{\rho_s}$ and setting $v=(0.3536-2\sqrt{\rho_s})\sqrt{n_1+n_2},$  and $k=\rho_s^{1/3}(n_1+n_2)^{1/6}$, we have
\begin{align*}
&P(\max_i |\lambda_i(\A)|>0.3536\sqrt{n_1+n_2})\\ \leq& \sqrt{n_1+n_2}\exp(-\frac{\rho_s^{1/3}(n_1+n_2)^{1/6}\cdot (0.3536-2\sqrt{\rho_s})\sqrt{n_1+n_2}}{0.3536\sqrt{n_1+n_2}})\\
\leq& (n_1+n_2)^{-10} < n_{(1)}^{-10}
\end{align*}
In the above,  $v>0$ because $\rho_s < 0.03$ and the second inequality holds because $\frac{(n_1+n_2)^{1/6}}{\log (n_1+n_2)}>\frac{10.5}{\rho_s^{1/3}(1-5.6561\sqrt{\rho_s})}.$
Clearly,
\begin{equation}
\|\A\|=\sqrt{\|\A\A^*\|}=\sqrt{\left\|\left(\begin{array}{cc} \Eb\Eb^*&0\\0&\Eb^*\Eb\end{array}\right)\right\|}=\sqrt{\|\Eb\Eb^*\|}=\|\Eb\|
\end{equation}

Therefore, we have
$P(\|\Eb\|>0.5\sqrt{n_{(1)}}) < n_{(1)}^{-10}.$
\end{proof}

\subsection{Implications of Assumption \ref{rhosr}}
We summarize here some important implications of Assumption \ref{rhosr}.

\begin{remark}
By Assumption \ref{rhosr}(a)(b)(c), we have
\begin{equation}
\label{rhos1}
\begin{array}{ll}
\rho_s &\leq 1- 1.5\max\left\{{60\rho_r^{1/2}}, 11C_{01}\rho_r^{1/2}, 0.11\right\}\\
&\leq 1- 1.5\max\left\{{60\rho_r^{1/2}}, 11C_{01}\rho_r^{1/2}, \frac{11 \log^2 n_{(1)}}{n_{(2)}}\right\}\\
&< \left(1-\frac{1.5\max\{{60\rho_r^{1/2}}, 11C_{01}\rho_r^{1/2}, \frac{11 \log^2 n_{(1)}}{n_{(2)}}\}}{1.5\log n_{(1)}}\right)^{1.5 \log n_{(1)}}\\
&< \left(1-\frac{\max\{{60\rho_r^{1/2}}, 11C_{01}\rho_r^{1/2}, \frac{11 \log^2 n_{(1)}}{n_{(2)}}\}}{\log n_{(1)}}\right)^{1.3\lceil \log n_{(1)}\rceil}
\end{array}
\end{equation}

The third inequality holds because $0 < 1.5\max\left\{{60\rho_r^{1/2}}, 0.11\right\}\leq 1.5\max\left\{60/10^2, 0.11\right\} < 1$; and for fixed constant $b>1$, $(1-x/b)^b > 1 - x$ whenever $x<1$. 
%
The fourth inequality holds since $1.5\log n_{(1)}> 1.3 \lceil \log n_{(1)}\rceil$ for $n_{(1)}\geq 1024.$
\end{remark}


\begin{remark}
By Assumption \ref{rhosr}(b)(c), we have
\begin{equation}
\label{rhos2}
\begin{array}{ll}
\rho_s \leq 0.0156 \leq 1- \frac{250C_{01}\rho_r}{\log n_{(1)}}.
\end{array}
\end{equation}
This follows since $n_{(1)} \geq \exp(253.9618C_{01}\rho_r)$ gives $\frac{250C_{01}\rho_r}{\log n_{(1)}} \leq 0.9844$, and so $1- \frac{250C_{01}\rho_r}{\log n_{(1)}} \geq 0.0156.$
\end{remark}

\subsection{Proof of Lemma \ref{WLc}}
\label{proofwlc}

The proof uses the following three lemmas.

\begin{lemma}\cite[Theorem 4.1]{candes2009exact}\cite[Theorem 2.6]{rpca}
  \label{POPTc0}
Suppose $\Omega_0 \sim  \text{Ber}(\rho_0)$.
  Then there is a numerical constant $C_{01}$ such that for all $\beta>1$,
  \begin{equation}
    \label{eq:rudelson}
    \|\PP - \rho_0^{-1} \PP \POzero \PP \| \le \epsilon_0,
  \end{equation}
with probability at least $1-3n_{(1)}^{-{\beta}}$ provided that $\rho_0 \geq C_{01} \, \epsilon_0^{-2} \, \frac{\beta\rho_r}{\log n_{(1)}}$.
\label{lem1}
\end{lemma}

\begin{lemma}\cite[Lemma 3.1]{rpca}
  \label{teo:infty}
  Suppose $\Z \in \Pi$ is a fixed matrix, and $\Omega_0 \sim
  \text{Ber}(\rho_0)$. Then
  \begin{equation}
    \label{eq:infty}
    \|\Z - \rho_0^{-1} \PP \Pc_{\Omega_0} \Z\|_\infty \le \epsilon_0 \|\Z\|_\infty
  \end{equation}
with probability at least $1-2n_{(1)}^{-11}$, provided that $\rho_0 \geq 60\, \epsilon_0^{-2} \, \frac{\rho_r}{\log n_{(1)}}$.
\end{lemma}
This is the same as Lemma 3.1 in \cite{rpca} except that we derive an explicit expression for the lower bound on $\rho_0$.  A proof for this can be found in the Appendix \ref{vershynin_bound}. 

\begin{lemma}\cite[Theorem 6.3]{candes2009exact}\cite[Lemma 3.2]{rpca}
\label{teo:sixthree}
Suppose $\Z$ is fixed, and $\Omega_0 \sim \text{Ber}(\rho_0)$. Then there is a constant $C_{03} > 0$ s.t.
  \label{teo:CR}
   \begin{equation}
    \label{eq:CR}
    \|(\Ib- \rho_0^{-1} \POzero) \Z\| \le C_{03} \sqrt{\frac{11 n_{(1)} \log n_{(1)}}{\rho_0}}
    \|\Z\|_\infty
  \end{equation}
  with probability at least $1-n_{(1)}^{-11}$, provided that $\rho_0 \geq \frac{
    11\log n_{(1)}}{n_{(2)}}$.
\end{lemma}

In the following proof, we take
\begin{equation}
\label{epsilon}
\epsilon= (\rho_r)^{1/4} \ \text{and} \ q=1-\rho_s^{\frac{1}{1.3\lceil \log n_{(1)}\rceil}}
\end{equation}
Notice from our assumption on $\rho_r$ given in Assumption \ref{assrhosrhors} that
 $$\epsilon \le (10^{-4})^{1/4}  \leq e^{-1}.$$


Let $\Z_j=\U_\new \V_\new^*-\PP \Y_j$. Clearly, $\Z_j \in \Pi$.  From the definition of $\Y_j$, notice that $\Y_j \in \Omega^\perp$,
 $$\Y_j=\Y_{j-1}+q^{-1}\POj \Z_{j-1}, \ \text{and}$$
$$\Z_j=(\PP-q^{-1}\PP \Pc_{\bOmega_j}\PP)\Z_{j-1}.$$
Clearly, $\bOmega_j$ and $\Z_{j-1}$ are independent. Using (\ref{rhos1}) and (\ref{epsilon}), $q\geq \frac{60\sqrt{\rho_r}}{\log n_{(1)}}$. Thus, by Lemma \ref{teo:infty}
\begin{equation}
\label{zj_infty}
\|\Z_j\|_\infty\leq \epsilon^j \|\U_\new \V_\new^*\|_\infty,
\end{equation}
with probability at least $1-2jn_{(1)}^{-11}$.
By Lemma \ref{POPTc0} and $q\geq \frac{11C_{01}\sqrt{\rho_r}}{\log n_{(1)}}$, which follows from (\ref{rhos1}),
\begin{equation}
\label{zj_f}
\|\Z_j\|_F \leq \epsilon \|\Z_{j-1}\|_F\leq \epsilon^j\|\U_\new \V_\new^*\|_F=\epsilon^j\sqrt{r}
\end{equation}
with probability at least $1-3jn_{(1)}^{-11}$.

%
{\bf Proof of (a)}
\begin{proof}
As
\begin{equation}
\label{Yj0}
\Y_{j_0}=\sum_{j=1}^{j_0} q^{-1}\Pc_{\bOmega_j}\Z_{j-1},
\end{equation}
and $\PPp \Z_j=0$, so we have, with probability at least $1- 3j_0n_{(1)}^{-11}$,
\begin{align*}
  \|\W^L\| =& \|\PPp \Y_{j_0}\| \le \sum_{j=1}^{j_0}  \|q^{-1} \PPp \POj \Z_{j-1}\|\\
  =& \sum_{j=1}^{j_0}  \|\PPp (q^{-1} \POj \Z_{j-1} - \Z_{j-1})\|\\
  \le& \sum_{j=1}^{j_0} \|q^{-1} \POj \Z_{j-1} - \Z_{j-1}\|\\
  \le& C_{03} \sqrt{\frac{11 n_{(1)}\log n_{(1)}}{q}} \sum_{j=1}^{j_0} \|\Z_{j-1}\|_\infty\\
     &\text{(using Lemma \ref{teo:sixthree} and $q\geq \frac{11\log n_{(1)}}{n_{(2)}}$ by (\ref{rhos1}))}\\
  \le&  C_{03} \sqrt{\frac{11 n_{(1)}\log n_{(1)}}{q}} \sum_{j=1}^{j_0} \epsilon^{j-1} \|\U_\new \V_\new^*\|_\infty\\
     &\text{(using Lemma \ref{teo:infty} and $q\geq \frac{60\rho_r^{1/2}}{\log n_{(1)}}$ by (\ref{rhos1}))}\\   
  <&  C_{03} (1-\epsilon)^{-1} \sqrt{\frac{11 n_{(1)} \log n_{(1)}}{q}} \|\U_\new \V_\new^*\|_\infty\\
  \le&  C_{03} (1-\epsilon)^{-1}\sqrt{\frac{11\rho_r}{q\log n_{(1)}}}\\
  & \text{(using $\|\U_\new \V_\new^*\|_\infty\leq \sqrt{\frac{\rho_r}{n_{(1)}\log^2 n_{(1)}}}$ by (\ref{eq:UV}))}\\
  \le&  \frac{\sqrt{11}C_{03}\rho_r^{1/4}} {\sqrt{60}(1-e^{-1})} \\
  & \text{(using $q\geq \frac{60\sqrt{\rho_r}}{\log n_{(1)}}$ by (\ref{rhos1}) and $\epsilon\leq e^{-1}$)}\\ 
  \leq&  \frac{1}{16}\\
  & \text{(using $\rho_r \leq 7.2483\times10^{-5}C_{03}^{-4}$ by Assu. \ref{assrhosrhors}(a))}
\end{align*}
The fourth step holds with probability at least $1-j_0n_{(1)}^{-11}$ by applying Lemma \ref{teo:sixthree} $j_0$ times; the fifth holds with probability at least $1-2j_0n_{(1)}^{-11}$ by applying Lemma \ref{teo:infty} $j_0$ times for each $\Z_j$ (similar to \eqref{zj_infty}).
Since $j_0 = 1.3 \log n_{(1)} < n_{(1)}$ (for $n_{(1)}$ satisfying Assumption  \ref{assrhosrhors}), the result follows.

\end{proof}

{\bf Proof of (b)}
\begin{proof}
Since $\PO \Y_{j_0}= 0$, we have $$\PO(\U_\new \V_\new^*+\PPp \Y_{j_0})=\PO(\U_\new \V_\new^* - \PP \Y_{j_0})= \PO(\Z_{j_0}),$$
and by (\ref{zj_f}), (\ref{epsilon}) and (\ref{rhos1}) ($q\geq \frac{11C_{01}\sqrt{\rho_r}}{\log n_{(1)}}$), we have
\begin{equation}
\label{zj0f}
\|\PO(\Z_{j_0})\|_F\le \|\Z_{j_0}\|_F \leq \epsilon^{j_0}\sqrt{r}\leq e^{-1.3\log n_{(1)}}\sqrt{r}=\frac{\sqrt{r}}{n_{(1)}^{1.3}},
\end{equation}
with probability at least $1-3{j_0}n_{(1)}^{-11}$. Thus, when $\frac{\sqrt{r}}{n_{(1)}^{0.8}}<\frac{1}{4}$, e.g. $n_{(1)}\geq 102$, Lemma \ref{WLc}(b) holds with probability at least $1-3n_{(1)}^{-10}$.  
\end{proof}

{\bf Proof of (c)}

\begin{proof}
Recall that $\U_\new \V_\new^*+\W^L=\Z_{j_0}+\Y_{j_0}$, $\POp \Y_{j_0}=\Y_{j_0}$. From above,
\begin{equation}
\label{zj0_infty}
\|\Z_{j_0}\|_\infty \leq \|\Z_{j_0}\|_F \leq \frac{\sqrt{r}}{n_{(1)}^{1.3}} < \frac{\lambda}{8}
\end{equation}
by (\ref{zj0f}) with probability at least $(1-3n_{(1)}^{-10})$ when $\frac{\sqrt{r}}{n_{(1)}^{0.8}}<\frac{1}{8}$, e.g. $n_{(1)}\geq 1024$. Thus, we only need to show $\|\Y_{j_0}\|_\infty \leq \frac{11\lambda}{40}$. We have, with probability at least $1-2j_0n_{(1)}^{-11}$,
\begin{equation}
\label{yj0_infty}
\begin{array}{lll}
\|\Y_{j_0}\|_\infty &\le & q^{-1} \sum_{j=1}^{j_0}  \|\POj \Z_{j-1}\|_\infty \\
  & \le &q^{-1} \sum_{j=1}^{j_0} \|\Z_{j-1}\|_\infty\\
  & \le &q^{-1} \sum_{j=1}^{j_0} \epsilon^{j-1} \|\U_\new \V_\new^*\|_\infty\\
  &&\text{(using Lemma \ref{teo:infty} and $q\geq \frac{60\rho_r^{1/2}}{\log n_{(1)}}$ by (\ref{rhos1}))}\\   
  & \le& q^{-1} \sum_{j=1}^{j_0} \epsilon^{j-1} \sqrt{\frac{\rho_r}{n_{(1)}\log^2 n_{(1)}}}\\
  && \text{(using $\|\U_\new \V_\new^*\|_\infty\leq \sqrt{\frac{\rho_r}{n_{(1)}\log^2 n_{(1)}}}$ by (\ref{eq:UV}))}\\
  & \le& \frac{\lambda}{60(1-e^{-1})} < \frac{11\lambda}{40} \\
  & &\text{(using $q\geq \frac{60\sqrt{\rho_r}}{\log n_{(1)}}$ by (\ref{rhos1}) and $\epsilon\leq e^{-1}$ by (\ref{epsilon}))}
\end{array}
\end{equation}
The third step follows from Lemma \ref{teo:infty} with probability at least $1-2j_0n_{(1)}^{-11}$. 
Thus, Lemma \ref{WLc}(c) holds with probability at least $1-2n_{(1)}^{-10}$.

To sum up, 
with the assumptions in Lemma \ref{WLc}, we have (a), (b), (c) of Lemma \ref{WLc} hold with probability at least $1-11n_{(1)}^{-10}$.
\end{proof}

\subsection{Proof of Lemma \ref{WSc}}
\label{proofwsc}

The proof uses the following lemma.

\begin{lemma}\cite[Corollary 2.7]{rpca}
  \label{POPTc}
  Assume that $\Omega_0 \sim \text{Ber}(\rho_0)$, $\Lb$ satisfies (\ref{eq:PU}), (\ref{eq:PV}) and (\ref{eq:UV}), then there is a numerical constant $C_{01}$ such that for all $\beta>1$,
   $$\|\Pc_{\Omega_0}\PP\|^2 \le
  \rho_0 + \epsilon_0,$$ with probability at least $1-3n_{(1)}^{-\beta}$ provided that $1-\rho_0 \geq C_{01} \, \epsilon_0^{-2} \, \frac{\beta\rho_r}{\log n_{(1)}}$.
\end{lemma}
This is a direct corollary of Lemma \ref{lem1} stated earlier. It follows by replacing $\Omega$ by $\Omega_0^c$ in Lemma \ref{lem1}.

{\bf Proof of (a)}

Let $\Eb := \sgn(\Sb)$. Recall from the assumption in this lemma that $\Eb$ satisfies the assumptions of Lemma \ref{sign2norm}.

By taking $\Omega_0 = \Omega$, $\rho_0=\rho_s$, $\epsilon_0=0.2,$ and $\beta =10$ in Lemma \ref{POPTc}, and using (\ref{rhos2}), we get
\begin{align}
\label{POPT}
\|\PO\PP\|^2 \le  \sigma:=\rho_s+0.2,
\end{align}
with probability at least $1-3n_{(1)}^{-10}$. Thus, using the bound on $\rho_s$ from (\ref{rhos2}), we get that $\|\PO\PP\|^2 \le 0.22 < 1/4$.

{\bf Proof of (b)}
\begin{proof}
Note that
\begin{align*}
  \W^S & = \PPp  \lambda \Eb + \PPp  \lambda  \sum_{k \ge 1}
  (\PO\PP\PO)^k \Eb\\
  & := \PPp \W_0^S + \PPp \W_1^S.
\end{align*}
By Assumption \ref{rhosr}(b)(e) and Lemma \ref{sign2norm}, we have
\[
\|\Eb\| \le 0.5\sqrt{n_{(1)} }
\]
with probability at least $1-n_{(1)}^{-10}$. Since $\lambda=1/\sqrt{n_{(1)}}$, we have
$$\|\PPp \W_0^S\| \le \|\W_0^S\| = \lambda\|\Eb\| \leq 0.5,$$
with probability at least $1-n_{(1)}^{-10}$.

Let $\mathcal{R} = \sum_{k \ge 1} (\PO\PP\PO)^{k}$. Let $N_1, N_2$ denote $1/2$-nets for $\mathbb{\Sb}^{n_1-1}, \mathbb{\Sb}^{n_2-1}$ where $\mathbb{\Sb}^{n_1-1}$ is a unit Euclidean sphere in $\mathbb{R}^{n_1}$.
A subset $N$ of $\mathbb{R}^{n_1}$ is referred to as a $\xi$-net, if and only if, for every $\y \in \mathbb{R}^{n_1}$, there is a $\y_1 \in N$ for which $\|\y-\y_1\| \le \xi$ (here we used the Euclidean distance metric) \cite{vershynin2010introduction}. 

By \cite[Lemma 5.2]{vershynin2010introduction}, the cardinality of the 1/2-nets $N_1$ and $N_2$ is $5^{n_1}$ and $5^{n_2}$ respectively.

By \cite[Lemma 5.4]{vershynin2010introduction},
\begin{eqnarray}
\|\mathcal{R}(\Eb)\| &=& \sup_{\x \in \mathbb{\Sb}^{n_2-1}, \y \in \mathbb{\Sb}^{n_1-1}} \<\y, \mathcal{R}(\Eb) \x\>  \nonumber \\
          & \le & 4 \sup_{\x\in N_2, \y \in N_1} \<\y, \mathcal{R}(\Eb) \x\>.
\end{eqnarray}
For a fixed pair $(\y,\x)$ of unit-normed vectors in $N_1 \times N_2$, define the random variable
\[
\X(\x,\y) := \<\y, \mathcal{R} (\Eb) \x\> = \<\mathcal{R}(\y\x^*), \Eb\>.
\]
Conditional on $\Omega = \text{supp}(\Eb)$, the signs of $\Eb$ are i.i.d.~symmetric and Hoeffding's inequality gives
\[
\Pb(|\X(\x,\y)| > t \, | \, \Omega) \le 2 \exp\Bigl(-\frac{2t^2}{
  \|\mathcal{R}(\y\x^*)\|_F^2}\Bigr).
\]
Now since $\|\y\x^*\|_F = 1$, the matrix $\mathcal{R}(\y\x^*)$ obeys $
\|\mathcal{R}(\y\x^*)\|_F \le \|\mathcal{R}\|$ and, therefore,
\[
\Pb\Bigl(\sup_{\x \in N_2, \y \in N_1} |\X(\x,\y)| > t \, | \, \Omega\Bigr) \le 2|N_1| |N_2|
\exp\Bigl(-\frac{2t^2}{\|\mathcal{R}\|^2}\Bigr).
\]
On the event $\{\|\PO\PP\| \le \sigma\}$, 
\[
\|\mathcal{R}\| \le \sum_{k \ge 1} \sigma^{2k} = \frac{\sigma^2}{1-\sigma^2}
\]
and, therefore, letting $\gamma =\frac{1-\sigma^2}{2\sigma^2}$, we have,
\[
\begin{array}{ll}
&\Pb(\lambda \|\mathcal{R}(\Eb)\| > \frac{27}{80}) \\
\le & \Pb(\lambda \|\mathcal{R}(\Eb)\| > \frac{27}{80}, \|\PO\PP\| \leq \sigma) + \Pb(\|\PO\PP\| > \sigma)\\
\le & \Pb\Bigl(\sup_{\x \in N_2, \y \in N_1} 4|\X(\x,\y)| > \frac{27\sqrt{n_{(1)}}}{80} \, | \, \|\PO\PP\| \leq \sigma\Bigr) +   \\& \Pb(\|\PO\PP\| > \sigma)\\
\le & 2|N_1| |N_2| \exp\Bigl(-\frac{27^2n_{(1)}\gamma^2}{12800}\Bigr) + \Pb(\|\PO\PP\| > \sigma)\\
\le &2 \times 5^{2n_{(1)}} \exp\Bigl(-\frac{27^2n_{(1)}\gamma^2}{12800}\Bigr) + 3n_{(1)}^{-10}\\
\le &2\exp\Bigl(-n_{(1)}(0.0570\gamma^2-\log 25)\Bigr) + 3n_{(1)}^{-10}\\
    & \text{(as $\sigma=\rho_s+0.2\leq 0.2156, \Rightarrow$ $0.0570\gamma^2-\log 25\geq 2.7773$)}\\
\le & 5n_{(1)}^{-10} \text{ (when $2.7773n_{(1)}\geq 10\log n_{(1)}$, e.g., $n_{(1)}\geq 10$.)}
\end{array}
\]
Thus
\[
\|\W^S\| \le 67/80,
\]
with probability at least $1-5 n_{(1)}^{-10}$. 
\end{proof}

{\bf Proof of (c)}
\begin{proof}
Observe that
\[
\begin{array}{ll}
 \POp \W^S  &= \lambda \POp (\I-\PP) (\PO - \PO \PP \PO)^{-1} \Eb\\
&= - \lambda \POp\PP (\PO - \PO \PP \PO)^{-1} \Eb
\end{array}
\]

Let $\W_3^S:= \POp \W^S$. Clearly, for $(i,j) \in \Omega$, $(\W_3^S)_{i,j} = 0$ and for $(i,j) \in \Omega^c$, $(\W_3^S)_{i,j} = (- \lambda \PP (\PO - \PO \PP \PO)^{-1} \Eb)_{i,j}$.

For $(i,j) \in \Omega^c$, it can be rewritten as
\[
\begin{array}{ll}
(\W_3^S)_{ij} &= \<\e_i, \W_3^S \e_j\> = \<\e_i \e_j^*, \W_3^S\>\\
&= \<\e_i \e_j^*, -\lambda \PP \PO (\PO - \PO \PP \PO)^{-1} \Eb\> \\
&=\lambda \<\X(i,j), \Eb\>
\end{array}
\]
where $\X(i,j):= -(\PO - \PO \PP \PO)^{-1} \PO \PP(\e_i\e_j^*)$.  Conditional on $\Omega = \text{supp}(\Eb)$, the signs
of $\Eb$ are i.i.d.~symmetric, and Hoeffding's inequality gives
\[
\Pb( |(\W_3^S)_{ij}| > t \lambda \, | \, \Omega) \le 2
\exp\Bigl(-\frac{2t^2}{ \|\X(i,j)\|_F^2}\Bigr),
\]
and, thus,
\[
\Pb\Bigl(\sup_{i,j \in \Omega^c} |(\W_3^S)_{ij}| > t \lambda \, | \, \Omega\Bigr) \le 2n_{1}n_{2}
\exp\Bigl(-\frac{2t^2}{\sup_{i,j} \|\X(i,j)\|_F^2}\Bigr).
\]
Since \eqref{eq:PPeiej} holds, on the event $\{\|\PO\PP\| \le \sigma\}$, we have
\[
\|\PO \PP(\e_i\e_j^*)\|_F \le \|\PO \PP\| \|\PP(\e_i\e_j^*)\|_F \le \sigma
\sqrt{2\rho_r/\log^2n_{(1)}}
\]
On the same event, $\|(\PO -
\PO \PP \PO)^{-1}\| \le (1-\sigma^2)^{-1}$ and, therefore,
\[
\|\X(i,j)\|^2_F \le \frac{2\sigma^2}{(1-\sigma^2)^2} \, \frac{\rho_r}{\log^2n_{(1)}}.
\]
Then unconditionally, letting $\gamma = \frac{(1-\sigma^2)^2}{2\sigma^2}$, we have
\[
\begin{array}{ll}
&\Pb\Bigl(\| \POp \W^S\|_\infty > \frac{\lambda}{2}\Bigr) = \Pb\Bigl(\| \W_3^S\|_\infty > \frac{\lambda}{2}\Bigr) \\
& \le  2n_{(1)}n_{(2)}\exp\left(-\frac{\log^2 n_{(1)}\gamma^2}{4\rho_r}\right)+ \Pb(\|\PO\PP\| \ge \sigma)\\
& \le 2n_{(1)}^{-\frac{\log n_{(1)} \gamma^2}{4\rho_r}+2} + 3n_{(1)}^{-10}\\
& \le 5n_{(1)}^{-10}
\end{array}
\]
The last bound follows since $\sigma=\rho_s+0.2\leq 0.2156$ by (\ref{rhos2}) and so $\gamma\geq 9.7798$; and $n_{(1)}\geq \exp(0.5019\rho_r)$ by Assumption \ref{rhosr}(c).

To sum up, with the assumption in Lemma \ref{WSc}, we have (a), (b) in Lemma \ref{WSc} hold with probability at least $1-10n_{(1)}^{-10}$.

\end{proof}

\subsection{Proof of Lemma \ref{FurediLemma}}
\label{proofFurediLemma}
\begin{proof}
The proof is the same as that given in \cite[Section 2]{vu2005spectral}. We rewrite it to clarify that variance of $a_{i,j}$ bounded by $\sigma^2$ also works.

As we know
$$\sum_{i=1}^n\lambda_i(\A)^k=\text{Trace}(\A^k),$$
we have
$$\sum_{i=1}^n\E(\lambda_i(\A)^k)=\E(\text{Trace}(\A^k)).$$
When $k$ is even, $\lambda_i(\A)^k$ are non-negative. Thus
$$\E(\max_i(|\lambda_i(\A)|^k) \leq \sum_{i=1}^n\E(\lambda_i(\A)^k) = \E(\text{Trace}(\A^k)).$$
Notice that
\begin{equation}
\text{Trace}\A^k = \sum_{i_1=1}^n\cdots\sum_{i_k=1}^n a_{i_1i_2}a_{i_2i_3}\cdots a_{i_{k-1}i_k}a_{i_ki_1},
\end{equation}
so we have
\begin{equation}
\E(\text{Trace}\A^k) = \sum_{i_1=1}^n\cdots\sum_{i_k=1}^n \E a_{i_1i_2}a_{i_2i_3}\cdots a_{i_{k-1}i_k}a_{i_ki_1}.
\end{equation}
For $1\leq p \leq k$, denote by $E(n,k,p)$ the sum of $\E a_{i_1i_2}a_{i_2i_3}\cdots a_{i_{k-1}i_k}a_{i_ki_1}$ over all sequences $i_1, i_2, \cdots, i_k$ such that $|\{i_1, i_2, \cdots, i_k\}|=p$ (i.e., $p$ different indices). As the $\E a_{ij}=0$, if some $a_{ij}$ in the product $a_{i_1i_2}a_{i_2i_3}\cdots a_{i_{k-1}i_k}a_{i_ki_1}$ has multiplicity one, then the expectation of the whole product is 0. When $p>(k/2)+1$, by pigeon hole principle, there must exist an $a_{ij}$ with multiplicity one. Thus $E(n,k,p)=0$ when $p>(k/2)+1$.

Note that a product $a_{i_1i_2}a_{i_2i_3}\cdots a_{i_{k-1}i_k}a_{i_ki_1}$ defines a closed walk
$$(i_1i_2)(i_2i_3)\cdots(i_{k-1}i_k)(i_ki_1)$$
of length $k$ on the complete graph $K_n$ on $\{1, \cdots, n\}$ (here we allow loops in $K_n$). If a product is non-zero, then any edge in the walk should appear at least twice. Denote by $W(n,k,p)$ the number of walks in $K_n$ using $k$ edges and $p$ vertices where each edge in the walk is used at least twice.

For a walk $W$ with $p$ vertices, denote by $V(W)=v_1, v_2, \cdots, v_p$ the ordered sequence. For graph $K_n$ with $n$ vertices, there are $n(n-1)\cdots(n-p+1)$ different ordered sequence. Denote by $W'(n,k,p)$ the number of walks with fixed sequence. Clearly, $$W(n,k,p)=n(n-1)\cdots(n-p+1)W'(n,k,p).$$

\begin{lemma}\cite{furedi1981eigenvalues}\cite[Lemma 2.1]{vu2005spectral}\cite[Problem 1.33]{lovasz1993combinatorial}
We have $$W'(n,k,p)\leq \left(k\atop 2p-2\right)p^{2(k-2p+2)}2^{2p-2}.$$
\end{lemma}
As $|a_{ij}|\leq K$, we have, for any $l\geq2$, $$\E(|a_{ij}|^l)\leq K^{l-2}\E(|a_{ij}|^2)\leq K^{l-2}\sigma^2.$$
With $p$ vertices, there are at least $p-1$ different $a_{ij}$'s, denoted by $\{a_{i_1j_1}, a_{i_2j_2}, \cdots, a_{i_{m}j_{m}}\}, m\geq p-1$, and each of them has multiplicity at least $2$, so we have
\begin{align*}
&\E(a_{i_1i_2}a_{i_2i_3}\cdots a_{i_{k-1}i_k}a_{i_ki_1}) \\
=&\E(a_{i_1j_1}^{l_1} a_{i_2j_2}^{l_2} \cdots a_{i_{m}j_{m}}^{l_m})\\
\leq& K^{k-(2p-2)}\E(a_{i_1j_1}^{2} a_{i_2j_2}^{2} \cdots a_{i_{p-1}j_{p-1}}^{2})\\
\leq& K^{k-(2p-2)}\sigma^{2p-2}
\end{align*}
Thus, we have
\[
\begin{array}{ll}
&E(n,k,p)\\
\leq& \sigma^{2p-2}K^{k-(2p-2)} W(n,k,p) \\
\leq& \sigma^{2p-2}K^{k-(2p-2)}n(n-1)\cdots(n-p+1)\left(k\atop 2p-2\right)p^{2(k-2p+2)}2^{2p-2}\\
\equiv& S(n,k,p)
\end{array}
\]
And
\begin{align*}
&\ \frac{S(n,k,p-1)}{S(n,k,p)}\\
=&\ \frac{K^2}{4\sigma^2(n-p+1)}\frac{\left(k\atop 2p-4\right)}{\left(k\atop 2p-2\right)}\frac{(p-1)^{2(k-2p+4)}}{p^{2(k-2p+2)}}\\
=&\ \frac{K^2}{4\sigma^2(n-p+1)}\frac{(2p-3)(2p-4)}{(k-2p+3)(k-2p+4)}\frac{(p-1)^{2(k-2p+4)}}{p^{2(k-2p+2)}}\\
\leq&\  \frac{K^2}{4\sigma^2n}\frac{k^2}{1}\frac{p^{2(k-2p+4)}}{p^{2(k-2p+2)}}\ \ \text{(because $p\leq k/2+1$)}\\
\leq&\  \frac{K^2k^6}{4\sigma^2n}
\end{align*}
Thus for $k\leq (\frac{\sigma}{K})^{1/3}(2n)^{1/6}$, $S(n,k,p-1)\leq \frac{1}{2}S(n,k,p)$. So
\begin{align*}
\E (\text{Trace}(\A^k)) &=\sum_{p=1}^{k/2+1}E(n,k,p)\\
&\leq \sum_{p=1}^{k/2+1}S(n,k,p)\\
&\leq 2S(n,k,k/2+1)\\
&=2\sigma^kn(n-1)\cdots(n-k/2)2^k\\
&\leq 2n(2\sigma\sqrt{n})^k
\end{align*}
By Markov's inequality, we have
\begin{align*}
&\ \Pb(\max_i(|\lambda_i(\A)|)\geq 2\sigma\sqrt{n}+v)\\
=&\ \Pb(\max_i(|\lambda_i(\A)|^k)\geq (2\sigma\sqrt{n}+v)^k)\\
\leq&\  \frac{\E(\max_i(|\lambda_i(\A)|^k))}{(2\sigma\sqrt{n}+v)^k}\\
\leq&\  \frac{2n(2\sigma\sqrt{n})^k}{(2\sigma\sqrt{n}+v)^k}\\
=&\  2n(1-\frac{v}{2\sigma\sqrt{n}+v})^k\\
\leq&\  2n\exp(-\frac{kv}{2\sigma\sqrt{n}+v})
\end{align*}
The last inequality holds for $0<\frac{v}{2\sigma\sqrt{n}+v}<1$, i.e., $v>0$. (Because for $0<x<1$, $(1-x)^k\leq \exp(-kx)$ $\Leftrightarrow$ $1-x \leq \exp(-x)$, which is easy to check. )
\end{proof}

\subsection{Bound on $\|\Eb\|$ by \cite{vershynin2010introduction}}
\label{vershynin_bound}

In \cite{rpca}, they need $\|\Eb\|<0.25\sqrt{n_{(1)}}$ with large probability. Here we derive the condition needed for $\|\Eb\|<\alpha\sqrt{n_{(1)}}, 0<\alpha<1$, with large probability.

By \cite[Lemma 5.36]{vershynin2010introduction}, and assume $\delta=\frac{\alpha}{\sqrt{\rho_s}}-1>1$, we only need to prove
$$\|\frac{1}{n_1\rho_s}\Eb^*\Eb-\Ib\|\leq \max(\delta,\delta^2)=\delta^2$$
with required probability.
By \cite[Lemma 5.4]{vershynin2010introduction}, for a $\frac{1}{4}$-net $\mathcal{N}$ of the unit sphere $S^{n-1}$, we have
\begin{align*}
&\ \|\frac{1}{n_1\rho_s}\Eb^*\Eb-\Ib\| \\
\leq &\ 2\max_{x\in\N}|\langle(\frac{1}{n_1\rho_s}\Eb^*\Eb-\Ib)x,x\rangle| \\
= &\ 2\max_{x\in \N}|\frac{1}{n_1\rho_s}\|\Eb x\|^2-1|.
\end{align*}
Thus we only need to prove $$\max_{x\in \N}|\frac{1}{n_1\rho_s}\|\Eb x\|^2-1|\leq \frac{\delta^2}{2}$$ with required probability.
By \cite[Lemma 5.2]{vershynin2010introduction}, we can choose the net $\N$ so that it has cardinality $|\N|\leq 9^{n_2}$.

As we know, for any unit norm vector $\x\in C^{n_2}$ and any fixed $\rho_s\in(0,1)$, $\{\frac{\Eb_i\x}{\sqrt{\rho_s}}\}_{i=1}^{n_1}$ are bounded by $\frac{\sum_{j=1}^{n_1}|\x_j|}{\sqrt{\rho_s}}$, thus they are sub-gaussian.
By \cite[Lemma 5.14]{vershynin2010introduction}, we have $\{\frac{|\Eb_i\x|^2}{\rho_s}\}_{i=1}^{n_1}$ are sub-exponential.
As $$\E \frac{|\Eb_i\x|^2}{\rho_s}= \|\x\|^2=1, i=1,2,\cdots,n_1,$$ thus by \cite[Remark 5.18]{vershynin2010introduction}, $\{\frac{|\Eb_i\x|^2}{\rho_s}-1\}_{i=1}^{n_1}$ are independent centered sub-exponential random variables and $\|\frac{|\Eb_i\x|^2}{\rho_s}-1\|_{\psi_1}\leq 2K_x$, where
$$K_x=\sup_{p\geq 1} p^{-1}(\E\frac{|\Eb_i\x|^{2p}}{\rho_s})^{1/p},$$
i.e., $$(\E\frac{|\Eb_i\x|^{2p}}{\rho_s})^{1/p}\leq K_xp, \ \forall p\geq 1,$$
Defined by \cite[(5.15)]{vershynin2010introduction}.

Let $$B_i=\frac{|\Eb_i\x|^2}{\rho_s}-1,\ \ i=1, 2, \cdots, n_1,$$ then $$\E B_i=0, (\E B_i^p)^{1/p}\leq 2K_xp, \forall p\geq1$$ and for $t\leq \frac{1}{4eK_x}$, we have

\begin{align*}
\E \exp(tB_i)&=1 + t\E B_i + \sum_{p=2}^\infty \frac{t^p\E B_i^{p}}{p!}\\
&\leq 1+\sum_{p=2}^\infty \frac{t^p2^pK_x^pp^p}{p!}\\
&\leq 1+\sum_{p=2}^\infty (2etK_x)^p\\
&\leq 1 + (2etK_x)^2\\
&\leq \exp(4e^2t^2K_x^2)
\end{align*}
the second inequality holds because $p!\geq (p/e)^p$; the third inequality holds because $2etK_x\leq 1/2$. Thus
 $$\E\exp(t\sum_{i=1}^{n_1}B_i)\leq \exp(4n_1e^2t^2K_x^2).$$
By Markov inequality, we have
\begin{align*}
\Pb(\frac{1}{n_1}\sum_{i=1}^{n_1}B_i \geq \frac{\delta^2}{2})&= \Pb(\exp(\frac{\tau}{n_1}\sum_{i=1}^{n_1}B_i )\geq \exp(\tau\delta^2/2))\\
&\leq e^{-\tau \delta^2/2}\E \exp(\frac{\tau}{n_1}\sum_{i=1}^{n_1}B_i )\\
&\leq e^{-\tau \delta^2/2 + 4e^2\tau^2K_x^2/n_1}
\end{align*}
when $\frac{\tau}{n_1}\leq \frac{1}{4eK_x}$, i.e., $\tau \leq \frac{n_1}{4eK_x}$.
Take $\tau=\min\{\frac{n_1\delta^2}{16e^2K_x^2},\frac{n_1}{4eK_x}\}$, we have
\begin{align*}
&\ \ \Pb(|\frac{1}{n_1\rho_s}\|\Eb \x\|^2-1|\geq \frac{\delta^2}{2})\\
&=\Pb(\frac{1}{n_1}\sum_{i=1}^{n_1}B_i \geq \frac{\delta^2}{2})\\
&\leq \exp({-\tau \delta^2/2 + 4e^2\tau^2K_x^2/n_1})\\
&\leq \exp({-\tau \delta^2/2 + \tau\delta^2/4})\\
&\leq \exp(-\min\{\frac{n_1\delta^4}{64e^2K_x^2},\frac{n_1\delta^2}{16eK_x}\})\\
&= \exp(-\frac{n_1\delta^2}{16eK_x}\min\{\frac{\delta^2}{4eK_x},1\})
\end{align*}
Let $$K=\sup_{\x\in \N} K_x,$$ then $$\Pb(\max_{\x\in \N}|\frac{1}{n_1\rho_s}\|\Eb \x\|^2-1|\geq \frac{\delta^2}{2})\leq 9^{n_2} \exp(-\frac{n_1\delta^2}{16eK}\min\{\frac{\delta^2}{4eK},1\}),$$
where $\delta^2=(\frac{\alpha}{\sqrt{\rho_s}}-1)^2=\frac{(\alpha-\sqrt{\rho_s})^2}{\rho_s}$.

So far the loose bound on $K$ we can get is $n_2/\rho_s$, so the best we can get is
\begin{align*}
&\ \ \Pb(\max_{\x\in \N}|\frac{1}{n_1\rho_s}\|\Eb \x\|^2-1|\geq \frac{\delta^2}{2})\\
&\leq 9^{n_2} \exp(-\frac{(\alpha-\sqrt{\rho_s})^2n_1}{16en_2}\min\{\frac{(\alpha-\sqrt{\rho_s})^2}{4en_2},1\})\\
&=9^{n_2} \exp\left(-\frac{(\alpha-\sqrt{\rho_s})^4n_1}{64e^2n_2^2}\right).
\end{align*}
Together with \cite[Lemma 5.36]{vershynin2010introduction}, we can get bound on $\|\Eb\|$. If we take $n_2=c\log n_1$ for some constant $c$, we have
\begin{align*}
&\ \ \Pb(\|\Eb\| \leq \alpha \sqrt{n_1})\\
 &= \Pb(\max_{\x\in \N}|\frac{1}{n_1\rho_s}\|\Eb \x\|^2-1|\geq \frac{\delta^2}{2})\\
 &\leq \exp\left(-\frac{(\alpha-\sqrt{\rho_s})^4n_1}{64e^2 c^2\log^2n_1}+c\log9\log n_1\right),
\end{align*}
which gives what we want when $n_1$ is large enough,
 $$-\frac{(\alpha-\sqrt{\rho_s})^4n_1}{64e^2 c^2\log^2n_1}+c\log9\log n_1 \leq - 10 \log n_1,$$
 i.e., $$\frac{n_1}{\log^3 n_1} \geq \frac{(\alpha-\sqrt{\rho_s})^4}{64e^2(10+c\log 9) c^2}$$
 But if $n_2$ is the order of $n_1$ or larger, we don't have the result with large probability.

\bibliographystyle{IEEEbib}
\bibliography{tipnewpfmt_kfcsfullpap}

\end{document}